\documentclass{llncs}
\usepackage[linesnumbered,ruled,noend]{algorithm2e}
\usepackage{float}
\usepackage{mdframed}
\usepackage{fancyvrb}
\usepackage{amssymb}
\usepackage{amsmath}
\usepackage{hyperref}
\usepackage{tikz}
\usetikzlibrary{calc}
\usetikzlibrary{quantikz2}
\usetikzlibrary {graphs,quotes}
\usetikzlibrary{arrows}
\usetikzlibrary{positioning}
\usetikzlibrary{cd}
\usetikzlibrary{patterns.meta}
\usepackage{caption}
\usepackage{subcaption}
\usepackage{textcomp}
\usepackage{graphics}
\usepackage[capitalize]{cleveref}
\usepackage{paralist}
\usepackage{xcolor}
\usepackage{booktabs}
\usepackage{wrapfig}
\usepackage{marvosym}
\usepackage{graphicx}
\usepackage{empheq}
\usepackage{proof}
\usepackage{listings}
\usepackage{xcolor}
\usepackage{nicefrac}
\usepackage{lineno}
\usepackage{parallel}
\usepackage{syntax}

\usetikzlibrary{arrows}
\usetikzlibrary{positioning}
\usetikzlibrary{graphs}
\usetikzlibrary{backgrounds}
\usepackage{forest}
\usepackage{newtxtext}
\usepackage{newtxmath}
\usepackage{physics}
\newcommand{\tikztransone}[5]{\text{\small {
  \protect\tikz[baseline,node distance=2mm]{%
  \pgfsetlinewidth{0.9bp}
  \tikzstyle{bddnode}=[st,fill=white,rounded corners=2mm]
  \protect\node[bddnode] (top) at (0,.64ex) {\hspace{.1em}\texttt{\upshape\strut$#1$}\hspace{.1em}\strut};
  \protect\node[right=of top.east,xshift=-2mm,yshift=1.5mm] (ch) {\scriptsize $#2$};
  \protect\node[bddnode,opacity=0,text opacity=1,right=of ch.east,xshift=-1mm,yshift=-1.5mm] (sym) {$#3$};
  \protect\node[bddnode,right=of sym.east,xshift=-1mm] (ql) {\hspace{.1em}\texttt{\upshape\strut$#4$}\hspace{.1em}\strut};
  \protect\draw (sym.west) coordinate[xshift=-1mm] (topa);
  \protect\draw (sym.north west) coordinate[xshift=3mm,yshift=0.5mm] (symtop);
  \protect\draw (sym.south west) coordinate[xshift=3mm,yshift=-0.5mm] (symbot);
  \protect\draw (sym.north east) coordinate[xshift=-2mm,yshift=0.4mm] (syma);
  \protect\draw (sym.south east) coordinate[xshift=-2mm,yshift=-0.4mm] (symb);
  \protect\draw (ql.north) coordinate[xshift=-1mm,yshift=0.5mm] (qla);
  \protect\draw[dashed,->,>=stealth'] (topa) to[bend left] (symtop)
  to (qla) to (ql.north);
  \protect\draw[->,>=stealth'] (top) to (topa) to[bend right] (symbot) %to (qlb)
  to[bend right=10] (ql.south);
  \protect\filldraw[fill=blue,opacity=0.2] (topa) to[bend left] (symtop) to (syma) to[bend left=50] (symb) to (symbot) to[bend left] cycle;
  }}}}
\newcommand{\tikztrans}[5]{\text{\small {
  \protect\tikz[baseline,node distance=2mm]{%
  \pgfsetlinewidth{0.9bp}
  \tikzstyle{bddnode}=[st,fill=white,rounded corners=2mm]
  \protect\node[bddnode] (top) at (0,.64ex) {\hspace{.1em}\texttt{\upshape\strut$#1$}\hspace{.1em}\strut};
  \protect\node[right=of top.east,xshift=-2mm,yshift=1.5mm] (ch) {\scriptsize $#2$};
  \protect\node[bddnode,opacity=0,text opacity=1,right=of ch.east,xshift=-1mm,yshift=-1.5mm] (sym) {$#3$};
  \protect\node[bddnode,right=of sym.east,xshift=-1mm] (ql) {\hspace{.1em}\texttt{\upshape\strut$#4$}\hspace{.1em}\strut};
  \protect\node[bddnode,right=of ql.east,xshift=-1mm] (qr) {\hspace{.1em}\texttt{\upshape\strut$#5$}\hspace{.1em}\strut};
  \protect\draw (sym.west) coordinate[xshift=-1mm] (topa);
  \protect\draw (sym.north west) coordinate[xshift=3mm,yshift=0.5mm] (symtop);
  \protect\draw (sym.south west) coordinate[xshift=3mm,yshift=-0.5mm] (symbot);
  \protect\draw (sym.north east) coordinate[xshift=-2mm,yshift=0.4mm] (syma);
  \protect\draw (sym.south east) coordinate[xshift=-2mm,yshift=-0.4mm] (symb);
  \protect\draw (ql.north) coordinate[xshift=-1mm,yshift=0.5mm] (qla);
  \protect\draw (qr.south) coordinate[xshift=-3mm,yshift=-0.5mm] (qra);
  \protect\draw[dashed,->,>=stealth'] (topa) to[bend left] (symtop)
  to (qla) to (ql.north);
  \protect\draw[->,>=stealth'] (top) to (topa) to[bend right] (symbot) to (qra)
  to[bend right=10] (qr.south);
  \protect\filldraw[fill=blue,opacity=0.2] (topa) to[bend left] (symtop) to (syma) to[bend left=50] (symb) to (symbot) to[bend left] cycle;
  }}}}
\newcommand{\tikzleaftrans}[3]{\text{\small {
  \protect\tikz[baseline,node distance=2mm]{%
  \pgfsetlinewidth{0.9bp}%
  \tikzstyle{bddnode}=[st,fill=white,rounded corners=2mm]%
  \protect\node[bddnode] (top) at (0,.64ex) {\hspace{.1em}\texttt{\upshape\strut$#1$}\hspace{.1em}\strut};%
  \protect\node[right=of top.east,xshift=-2mm,yshift=1.5mm] (ch) {\scriptsize $#2$};%
  \protect\node[bddnode,opacity=0,text opacity=1,right=of ch.east,xshift=-2mm,yshift=-1.5mm] (sym) {\texttt{\upshape\strut$#3$}\strut};%
  \protect\draw[-,>=stealth'] (top) to (sym);%
  }}}}
\newcommand{\lsta}[0]{LSTA\xspace}
\newcommand{\lstas}[0]{LSTAs\xspace}
\newcommand{\showcomment}[1]{#1}
\newcommand{\ol}[1]{\showcomment{\textcolor{blue}{\ifmmode \text{[OL: #1]}\else [OL: #1] \fi}}}
\newcommand{\yfc}[1]{\showcomment{\textcolor{purple}{\ifmmode \text{#1}\else #1 \fi}}}
\newcommand{\km}[1]{\showcomment{\textcolor{green}{\ifmmode \text{[KM: #1]}\else [KM: #1] \fi}}}
\newcommand{\wl}[1]{\showcomment{\textcolor{orange}{\ifmmode \text{[WL: #1]}\else [WL: #1] \fi}}}
\newcommand{\ja}[1]{\showcomment{\textcolor{red}{\ifmmode \text{[JA: #1]}\else [JA: #1] \fi}}}

\newcommand{\lnref}[1]{Line~{\ref{#1}}}
\newcommand{\lnrangeref}[2]{Lines~{\ref{#1}--\ref{#2}}}
\newcommand{\FF}{\ensuremath{\mathbf{false}}} 
\newcommand{\TT}{\ensuremath{\mathbf{true}}}
\newcommand{\aut}[0]{\mathcal{A}}
\newcommand{\autb}[0]{\mathcal{B}}

\newcommand{\autp}[0]{\mathcal{P}}
\newcommand{\auti}[0]{\mathcal{I}}
\newcommand{\autq}[0]{\mathcal{Q}}
\newcommand{\autr}[0]{\mathcal{R}}
\newcommand{\Q}[0]{q}
\newcommand{\R}[0]{r}

\newcommand{\vars}[0]{\mathbb{X}}
\newcommand{\varsY}[0]{\mathbb{Y}}
\newcommand{\varsZ}[0]{\mathbb{Z}}

\newcommand{\varsof}[1]{\mathit{vars}(#1)}

\newcommand{\topof}[1]{\mathtt{top}(#1)}
\newcommand{\symof}[1]{\mathtt{sym}(#1)}
\newcommand{\botof}[1]{\mathtt{bot}(#1)}
\newcommand{\leftof}[1]{\mathtt{left}(#1)}
\newcommand{\rightof}[1]{\mathtt{right}(#1)}

\newcommand{\merge}[0]{\mathop{\Cup}}

\newcommand{\formulaof}[1]{\tikz[anchor=base,baseline]{
  \node[inner sep=0.5mm, fill=blue!10, minimum height=4mm] {$#1$};}}

\newcommand{\scale}[0]{\mathit{scale}}
\newcommand{\setform}[0]{\Phi}
\newcommand{\globconstr}[0]{\varphi}
\newcommand{\issat}[0]{\mathit{IsSat}}
\newcommand{\issatof}[1]{\issat(#1)}

\tikzset{st/.style={font=\ttfamily,shape=rectangle,rounded corners=.5em,draw=black,fill=gray!30,inner xsep=.2em,inner ysep=0em,text height=2.2ex,text depth=1.0ex}}
\tikzset{leaf/.style={font=\ttfamily,shape=rectangle,rounded corners=.5em,draw=black,fill=blue!20,inner xsep=.2em,inner ysep=0em,text height=2.2ex,text depth=1.0ex}}

\newcommand{\trans}[3]{#1 \xrightarrow{#2} #3}
\newcommand{\transtree}[3]{\trans{#1}{#2}{(#3)}}
\newcommand{\transleaf}[2]{\transtree{#1}{#2}{}}

\newcommand{\worklist}[0]{\mathit{Worklist}}
\newcommand{\processed}[0]{\mathit{Processed}}

\newcommand{\hide}[1]{}

\newcommand{\autoqq}[0]{\textsc{AutoQ~2.0}\xspace}
\newcommand{\autoq}[0]{\textsc{AutoQ~1.0}\xspace}

\newcommand{\sliqec}[0]{\textsc{SliQEC}\xspace}
\newcommand{\feynman}[0]{\textsc{Feynman}\xspace}

\newcommand{\qcec}[0]{\textsc{QCEC}\xspace}
\newcommand{\qbricks}[0]{\textsc{Qbricks}\xspace}
\newcommand{\coqq}[0]{\textsc{CoqQ}\xspace}

\newcommand{\openqasm}[0]{\textsc{OpenQASM}\xspace}
\newcommand{\qpmc}[0]{\textsc{QPMC}\xspace}
\newcommand{\coq}[0]{\textsc{Coq}\xspace}

\newcommand{\isabelle}[0]{\textsc{Isabelle}\xspace}

\newcommand{\feasibletrans}[0]{\mathtt{FTrans}}
\newcommand{\feasibletransof}[2]{\feasibletrans(#1, #2)}

\newcommand{\tuple}[1]{\langle #1 \rangle}

\newcommand{\lang}[0]{\mathcal{L}}
\newcommand{\langof}[1]{\lang(#1)}

\newcommand{\entailedbyuptosc}[0]{\mathrel{\models^{\mathit{uts}}}}

\newcommand{\ziii}[0]{Z3\xspace}
\newcommand{\complex}[0]{\mathbb{C}}
\newcommand{\symterm}[0]{\textsf{term}(\complex, \vars)}

\newcommand{\specpair}[2]{#2\ket{#1}}
\newcommand{\specdef}[1]{#1\specpair{*}{} }

\newcommand{\measurement}[2]{M_{#1}{#2}}

\newcommand{\nat}[0]{\mathbb{N}}
\newcommand{\natz}[0]{\mathbb{N}_0}
\newcommand{\reals}[0]{\mathbb{R}}
\newcommand{\realsext}[0]{\reals^{?}}

\newcommand{\Arity}[0]{\ifmmode \mathbf{Arity} \else \textbf{Arity} \fi}

\newcommand{\run}[0]{\rho}
\newcommand{\rootstates}[0]{\mathcal{R}}
\newcommand{\ctranstreenoset}[4]{\transtree{#1}{#2\mid#4}{#3}}

\newcommand{\ctr}[2]{#1}
\newcommand{\dom}[0]{\mathrm{dom}}
\newcommand{\domof}[1]{\dom(#1)}
\newcommand{\img}[0]{\mathrm{img}}
\newcommand{\imgof}[1]{\img(#1)}

\lstdefinelanguage{OpenQASM}{
  keywords=[1]{qubit, bit, gate, barrier, measure, reset},
  keywords=[2]{cx, ccx, x, h, cu},
  sensitive=true,
  morecomment=[l]{//},
  morestring=[b]',
}

\lstdefinestyle{OpenQASMStyle}{
  language=OpenQASM,
  basicstyle=\small\ttfamily,
  keywordstyle=[1]\color{blue},
  keywordstyle=[2]\color{green!60!black},
  commentstyle=\color{gray}\itshape,
  stringstyle=\color{red},
  numbers=left,
  numberstyle=\color{gray},
  stepnumber=1,
  numbersep=15pt,
  backgroundcolor=\color{gray!10},
  frame=single,
  framesep=2mm,
  rulecolor=\color{black},
  breaklines=true,
  showstringspaces=false,
}

\newcommand{\tA}[0]{
\begin{tikzpicture}[yscale=1,xscale=1]
\node {$x_1$}[sibling distance = 2.4cm, level distance = 0.8cm]
child  {node {$x_2$} edge from parent [dashed, sibling distance = 1.2cm]
    child  {node (1) {$x_3$} edge from parent [dashed, level distance = 1.5cm, sibling distance = 1.2cm]
        child  {node (11) {} edge from parent [color=white]}
        child  {node (12) {} edge from parent [color=white]}
    }
    child {node (2) {$x_3$} edge from parent [solid, level distance = 1.5cm, sibling distance = 1.2cm]
        child  {node (21) {} edge from parent [color=white]}
        child  {node (22) {} edge from parent [color=white]}
    };
}
child {node {$x_2$} edge from parent [solid, sibling distance = 1.2cm]
    child  {node (3) {$x_3$} edge from parent [dashed, level distance = 1.5cm, sibling distance = 1.2cm]
        child  {node (31) {} edge from parent [color=white]}
        child  {node (32) {} edge from parent [color=white]}
    }
    child {node (4) {$x_3$} edge from parent [solid, level distance = 1.5cm, sibling distance = 1.2cm]
        child  {node (41) {} edge from parent [color=white]}
        child  {node (42) {} edge from parent [color=white]}
    }
};
\draw (1.south) -- (11.center);
\draw (1.south) -- (12.center);
\draw (11.center) -- (12.center);
\draw (2.south) -- (21.center);
\draw (2.south) -- (22.center);
\draw (21.center) -- (22.center);
\draw (3.south) -- (31.center);
\draw (3.south) -- (32.center);
\draw (31.center) -- (32.center);
\draw (4.south) -- (41.center);
\draw (4.south) -- (42.center);
\draw (41.center) -- (42.center);

\coordinate (A) at ($(1.center) + (265:0.7cm)$);
\coordinate (B) at ($(1.center) + (275:1.2cm)$);
\coordinate (midpoint1) at ($(11)!0.5!(12)$);
\node at ($(midpoint1.south) + (-0.04cm,-0.25cm)$) {$\ell\ ......\ \ell$};
\coordinate (midpoint2) at ($(21)!0.5!(22)$);
\coordinate (midpoint3) at ($(31)!0.5!(32)$);
\node at ($(midpoint3) + (0,-0.25cm)$) {$0\ ................................\ 0$};
\coordinate (midpoint4) at ($(41)!0.5!(42)$);
\end{tikzpicture}
}
\newcommand{\tB}[0]{
\begin{tikzpicture}[yscale=1,xscale=1]
\node {$x_1$}[sibling distance = 2.4cm, level distance = 0.8cm]
        child  {node {$x_2$} edge from parent [dashed, sibling distance = 1.2cm]
            child  {node (1) {$x_3$} edge from parent [dashed, level distance = 1.5cm, sibling distance = 1.2cm]
                child  {node (11) {} edge from parent [color=white]}
                child  {node (12) {} edge from parent [color=white]}
            }
            child {node (2) {$x_3$} edge from parent [solid, level distance = 1.5cm, sibling distance = 1.2cm]
                child  {node (21) {} edge from parent [color=white]}
                child  {node (22) {} edge from parent [color=white]}
            };
        }
        child {node {$x_2$} edge from parent [solid, sibling distance = 1.2cm]
            child  {node (3) {$x_3$} edge from parent [dashed, level distance = 1.5cm, sibling distance = 1.2cm]
                child  {node (31) {} edge from parent [color=white]}
                child  {node (32) {} edge from parent [color=white]}
            }
            child {node (4) {$x_3$} edge from parent [solid, level distance = 1.5cm, sibling distance = 1.2cm]
                child  {node (41) {} edge from parent [color=white]}
                child  {node (42) {} edge from parent [color=white]}
            }
        };
        \draw (1.south) -- (11.center);
        \draw (1.south) -- (12.center);
        \draw (11.center) -- (12.center);
        \draw (2.south) -- (21.center);
        \draw (2.south) -- (22.center);
        \draw (21.center) -- (22.center);
        \draw (3.south) -- (31.center);
        \draw (3.south) -- (32.center);
        \draw (31.center) -- (32.center);
        \draw (4.south) -- (41.center);
        \draw (4.south) -- (42.center);
        \draw (41.center) -- (42.center);

        \coordinate (A) at ($(1.center) + (265:0.7cm)$);
        \coordinate (B) at ($(1.center) + (275:1.2cm)$);
        \coordinate (midpoint1) at ($(11)!0.5!(12)$);
        \draw (1.south) -- (A) -- (B) -- (midpoint1.center);
        \node (C) [left=0cm of B, font=\scriptsize] {sol.};

        \coordinate (A2) at ($(2.center) + (265:0.7cm)$);
        \coordinate (B2) at ($(2.center) + (275:1.2cm)$);
        \coordinate (midpoint2) at ($(21)!0.5!(22)$);
        \draw (2.south) -- (A2) -- (B2) -- (midpoint2.center);
        \node (C2) [left=0cm of B2, font=\scriptsize] {sol.};
        
        \node at ($(midpoint1.south) + (0,-0.25cm)$) {$\ell...0...\ell$};
        \node at ($(midpoint2) + (0,-0.25cm)$) {$0...\ell...0$};
        \coordinate (midpoint3) at ($(31)!0.5!(32)$);
        \node at ($(midpoint3) + (0,-0.25cm)$) {$0\ .........$};
        \coordinate (midpoint4) at ($(41)!0.5!(42)$);
        \node at ($(midpoint4) + (0,-0.25cm)$) {$..........\ 0$};
\end{tikzpicture}
}
\newcommand{\tC}[0]{
\begin{tikzpicture}[yscale=1,xscale=1]
\node {$x_1$}[sibling distance = 2.4cm, level distance = 0.8cm]
        child  {node {$x_2$} edge from parent [dashed, sibling distance = 1.2cm]
            child  {node (1) {$x_3$} edge from parent [dashed, level distance = 1.5cm, sibling distance = 1.2cm]
                child  {node (11) {} edge from parent [color=white]}
                child  {node (12) {} edge from parent [color=white]}
            }
            child {node (2) {$x_3$} edge from parent [solid, level distance = 1.5cm, sibling distance = 1.2cm]
                child  {node (21) {} edge from parent [color=white]}
                child  {node (22) {} edge from parent [color=white]}
            };
        }
        child {node {$x_2$} edge from parent [solid, sibling distance = 1.2cm]
            child  {node (3) {$x_3$} edge from parent [dashed, level distance = 1.5cm, sibling distance = 1.2cm]
                child  {node (31) {} edge from parent [color=white]}
                child  {node (32) {} edge from parent [color=white]}
            }
            child {node (4) {$x_3$} edge from parent [solid, level distance = 1.5cm, sibling distance = 1.2cm]
                child  {node (41) {} edge from parent [color=white]}
                child  {node (42) {} edge from parent [color=white]}
            }
        };
        \draw (1.south) -- (11.center);
        \draw (1.south) -- (12.center);
        \draw (11.center) -- (12.center);
        \draw (2.south) -- (21.center);
        \draw (2.south) -- (22.center);
        \draw (21.center) -- (22.center);
        \draw (3.south) -- (31.center);
        \draw (3.south) -- (32.center);
        \draw (31.center) -- (32.center);
        \draw (4.south) -- (41.center);
        \draw (4.south) -- (42.center);
        \draw (41.center) -- (42.center);

        \coordinate (A) at ($(1.center) + (265:0.7cm)$);
        \coordinate (B) at ($(1.center) + (275:1.2cm)$);
        \coordinate (midpoint1) at ($(11)!0.5!(12)$);
        \draw (1.south) -- (A) -- (B) -- (midpoint1.center);
        \node (C) [left=0cm of B, font=\scriptsize] {sol.};

        \coordinate (A2) at ($(2.center) + (265:0.7cm)$);
        \coordinate (B2) at ($(2.center) + (275:1.2cm)$);
        \coordinate (midpoint2) at ($(21)!0.5!(22)$);
        \draw (2.south) -- (A2) -- (B2) -- (midpoint2.center);
        \node (C2) [left=0cm of B2, font=\scriptsize] {sol.};

        \coordinate (A4) at ($(4.center) + (265:0.7cm)$);
        \coordinate (B4) at ($(4.center) + (275:1.2cm)$);
        \coordinate (midpoint4) at ($(41)!0.5!(42)$);
        \draw (4.south) -- (A4) -- (B4) -- (midpoint4.center);
        \node (C4) [left=0cm of B4, font=\scriptsize] {sol.};
        
        \node at ($(midpoint1.south) + (0,-0.25cm)$) {$\ell...0...\ell$};
        \node at ($(midpoint2) + (0,-0.25cm)$) {$0..b\ell..0$};
        \coordinate (midpoint3) at ($(31)!0.5!(32)$);
        \node at ($(midpoint3) + (0,-0.25cm)$) {$0\ ......\ 0$};
        \coordinate (midpoint4) at ($(41)!0.5!(42)$);
        \node at ($(midpoint4) + (0,-0.25cm)$) {$0..a\ell..0$};
\end{tikzpicture}
}
\newcommand{\tD}[0]{
\begin{tikzpicture}[yscale=1,xscale=1]
\node {$x_1$}[sibling distance = 2.4cm, level distance = 0.8cm]
        child  {node {$x_2$} edge from parent [dashed, sibling distance = 1.2cm]
            child  {node (1) {$x_3$} edge from parent [dashed, level distance = 1.5cm, sibling distance = 1.2cm]
                child  {node (11) {} edge from parent [color=white]}
                child  {node (12) {} edge from parent [color=white]}
            }
            child {node (2) {$x_3$} edge from parent [solid, level distance = 1.5cm, sibling distance = 1.2cm]
                child  {node (21) {} edge from parent [color=white]}
                child  {node (22) {} edge from parent [color=white]}
            };
        }
        child {node {$x_2$} edge from parent [solid, sibling distance = 1.2cm]
            child  {node (3) {$x_3$} edge from parent [dashed, level distance = 1.5cm, sibling distance = 1.2cm]
                child  {node (31) {} edge from parent [color=white]}
                child  {node (32) {} edge from parent [color=white]}
            }
            child {node (4) {$x_3$} edge from parent [solid, level distance = 1.5cm, sibling distance = 1.2cm]
                child  {node (41) {} edge from parent [color=white]}
                child  {node (42) {} edge from parent [color=white]}
            }
        };
        \draw (1.south) -- (11.center);
        \draw (1.south) -- (12.center);
        \draw (11.center) -- (12.center);
        \draw (2.south) -- (21.center);
        \draw (2.south) -- (22.center);
        \draw (21.center) -- (22.center);
        \draw (3.south) -- (31.center);
        \draw (3.south) -- (32.center);
        \draw (31.center) -- (32.center);
        \draw (4.south) -- (41.center);
        \draw (4.south) -- (42.center);
        \draw (41.center) -- (42.center);

        \coordinate (A) at ($(1.center) + (265:0.7cm)$);
        \coordinate (B) at ($(1.center) + (275:1.2cm)$);
        \coordinate (midpoint1) at ($(11)!0.5!(12)$);
        \draw (1.south) -- (A) -- (B) -- (midpoint1.center);
        \node (C) [left=0cm of B, font=\scriptsize] {sol.};

        \coordinate (A3) at ($(3.center) + (265:0.7cm)$);
        \coordinate (B3) at ($(3.center) + (275:1.2cm)$);
        \coordinate (midpoint3) at ($(31)!0.5!(32)$);
        \draw (3.south) -- (A3) -- (B3) -- (midpoint3.center);
        \node (C3) [left=0cm of B3, font=\scriptsize] {sol.};
        
        \node at ($(midpoint1.south) + (0.075cm,-0.25cm)$) [text width=1.5cm] {$\ell...b\ell..\ell$};
        \coordinate (midpoint2) at ($(21)!0.5!(22)$);
        \node at ($(midpoint2) + (0,-0.25cm)$) {$0\ ......\ 0$};
        \coordinate (midpoint3) at ($(31)!0.5!(32)$);
        \coordinate (midpoint4) at ($(41)!0.5!(42)$);
        \node at ($(midpoint4) + (0,-0.25cm)$) {$0\ ......\ 0$};        
        \node at ($(midpoint3) + (0,-0.25cm)$) {$0..a\ell..0$};
\end{tikzpicture}
}
\newcommand{\tE}[0]{
\begin{tikzpicture}[yscale=1,xscale=1]
\node {$x_1$}[sibling distance = 2.4cm, level distance = 0.8cm]
        child  {node {$x_2$} edge from parent [dashed, sibling distance = 1.2cm]
            child  {node (1) {$x_3$} edge from parent [dashed, level distance = 1.5cm, sibling distance = 1.2cm]
                child  {node (11) {} edge from parent [color=white]}
                child  {node (12) {} edge from parent [color=white]}
            }
            child {node (2) {$x_3$} edge from parent [solid, level distance = 1.5cm, sibling distance = 1.2cm]
                child  {node (21) {} edge from parent [color=white]}
                child  {node (22) {} edge from parent [color=white]}
            };
        }
        child {node {$x_2$} edge from parent [solid, sibling distance = 1.2cm]
            child  {node (3) {$x_3$} edge from parent [dashed, level distance = 1.5cm, sibling distance = 1.2cm]
                child  {node (31) {} edge from parent [color=white]}
                child  {node (32) {} edge from parent [color=white]}
            }
            child {node (4) {$x_3$} edge from parent [solid, level distance = 1.5cm, sibling distance = 1.2cm]
                child  {node (41) {} edge from parent [color=white]}
                child  {node (42) {} edge from parent [color=white]}
            }
        };
        \draw (1.south) -- (11.center);
        \draw (1.south) -- (12.center);
        \draw (11.center) -- (12.center);
        \draw (2.south) -- (21.center);
        \draw (2.south) -- (22.center);
        \draw (21.center) -- (22.center);
        \draw (3.south) -- (31.center);
        \draw (3.south) -- (32.center);
        \draw (31.center) -- (32.center);
        \draw (4.south) -- (41.center);
        \draw (4.south) -- (42.center);
        \draw (41.center) -- (42.center);

        \coordinate (A) at ($(1.center) + (265:0.7cm)$);
        \coordinate (B) at ($(1.center) + (275:1.2cm)$);
        \coordinate (midpoint1) at ($(11)!0.5!(12)$);
        \draw (1.south) -- (A) -- (B) -- (midpoint1.center);
        \node (C) [left=0cm of B, font=\scriptsize] {sol.};
        
        \node at ($(midpoint1.south) + (0.075cm,-0.25cm)$) [text width=1.5cm] {$\ell...b\ell..\ell$};
        \coordinate (midpoint2) at ($(21)!0.5!(22)$);
        \coordinate (midpoint3) at ($(31)!0.5!(32)$);
        \coordinate (midpoint4) at ($(41)!0.5!(42)$);
        \node at ($(midpoint3) + (0.01cm,-0.25cm)$) {$0\ ................................\ 0$};
\end{tikzpicture}
}
\newcommand{\tF}[0]{
\begin{tikzpicture}[yscale=1,xscale=1]
\node {$x_1$}[sibling distance = 2.4cm, level distance = 0.8cm]
        child  {node {$x_2$} edge from parent [dashed, sibling distance = 1.2cm]
            child  {node (1) {$x_3$} edge from parent [dashed, level distance = 1.5cm, sibling distance = 1.2cm]
                child  {node (11) {} edge from parent [color=white]}
                child  {node (12) {} edge from parent [color=white]}
            }
            child {node (2) {$x_3$} edge from parent [solid, level distance = 1.5cm, sibling distance = 1.2cm]
                child  {node (21) {} edge from parent [color=white]}
                child  {node (22) {} edge from parent [color=white]}
            };
        }
        child {node {$x_2$} edge from parent [solid, sibling distance = 1.2cm]
            child  {node (3) {$x_3$} edge from parent [dashed, level distance = 1.5cm, sibling distance = 1.2cm]
                child  {node (31) {} edge from parent [color=white]}
                child  {node (32) {} edge from parent [color=white]}
            }
            child {node (4) {$x_3$} edge from parent [solid, level distance = 1.5cm, sibling distance = 1.2cm]
                child  {node (41) {} edge from parent [color=white]}
                child  {node (42) {} edge from parent [color=white]}
            }
        };
        \draw (1.south) -- (11.center);
        \draw (1.south) -- (12.center);
        \draw (11.center) -- (12.center);
        \draw (2.south) -- (21.center);
        \draw (2.south) -- (22.center);
        \draw (21.center) -- (22.center);
        \draw (3.south) -- (31.center);
        \draw (3.south) -- (32.center);
        \draw (31.center) -- (32.center);
        \draw (4.south) -- (41.center);
        \draw (4.south) -- (42.center);
        \draw (41.center) -- (42.center);

        \coordinate (A) at ($(1.center) + (265:0.7cm)$);
        \coordinate (B) at ($(1.center) + (275:1.2cm)$);
        \coordinate (midpoint1) at ($(11)!0.5!(12)$);
        \draw (1.south) -- (A) -- (B) -- (midpoint1.center);
        \node (C) [left=0cm of B, font=\scriptsize] {sol.};
        
        \node at ($(midpoint1.south) + (0.1cm,-0.25cm)$) [text width=1.5cm] {$\ell^-_{..}b\ell^+_{..}\ell^-$};
        \coordinate (midpoint2) at ($(21)!0.5!(22)$);
        \coordinate (midpoint3) at ($(31)!0.5!(32)$);
        \coordinate (midpoint4) at ($(41)!0.5!(42)$);
        \node at ($(midpoint3) + (0.01cm,-0.25cm)$) {$0\ ................................\ 0$};
\end{tikzpicture}
}

\newcommand{\inclusionExample}[0]{
\begin{wrapfigure}[11]{r}{0.55\textwidth}
\vspace{-7mm}
    \centering
    \scalebox{0.7}{
\vspace{-5mm}
    \begin{tikzpicture}[>=stealth',node distance=20mm]
  \pgfsetlinewidth{1bp}
  \tikzstyle{bddnode}=[draw,rectangle,rounded corners=2mm]
  \tikzstyle{bddleaf}=[]
  \tikzstyle{trans}=[->,>=stealth']
  \tikzstyle{translow}=[->,>=stealth',dashed]
  \tikzstyle{stick}=[-,>=stealth']
  \tikzstyle{link}=[->,dashed]
  \tikzstyle{hidtrans}=[]
  \tikzstyle{ark}=[]
  \tikzstyle{grayark}=[fill=black,opacity=0.2]
  \tikzstyle{blueark}=[fill=blue,opacity=0.2]
  \tikzstyle{redark}=[fill=red,opacity=0.6]

  \tikzstyle{outp}=[scale=0.75,fill=black!30,inner sep=0.6mm]

  \tikzstyle{bddnodex}=[bddnode,inner sep=1mm,minimum size=5mm]

  \node (root) {};

  \node[bddnodex,left of=root,xshift=-15] (Lp) {};
  \node[bddnodex,below left of=Lp,yshift=-5mm,xshift=3mm] (q1) {$q$};
  \node[bddnodex,below right of=Lp,yshift=-5mm,xshift=-3mm] (p_1) {$q$};
  \node[bddnodex,below left of=q1,yshift=-5mm,xshift=8mm] (r1) {$q^L$};
  \node[bddnodex,below right of=q1,yshift=-5mm,xshift=-8mm] (r0_1) {$q^R$};
  \node[bddnodex,below left of=p_1,yshift=-5mm,xshift=8mm] (r0_2) {$q^L$};
  \node[bddnodex,below right of=p_1,yshift=-5mm,xshift=-8mm] (r0_3) {$q^R$};
  \draw (Lp) coordinate[xshift=-0mm,yshift=-5mm] (Lpa);
  \draw (q1) coordinate[xshift=0mm,yshift=-5mm] (q1a);
  \draw (p_1) coordinate[xshift=0mm,yshift=-5mm] (p_1a);

  \node[bddnodex,right of=root,xshift=15] (Rp) {$\ \ $};
  \node[left of=Rp, xshift=10mm]  {$\autb, \rho_i$};
 \node[left of=Lp, xshift=10mm]  {$\aut, \rho$};  
  \node[bddnodex,below left of=Rp,yshift=-5mm,xshift=3mm] (r) {$r$};
  \node[bddnodex,below right of=Rp,yshift=-5mm,xshift=-3mm] (p_2) {$s$};
  \node[bddnodex,below left of=r,yshift=-5mm,xshift=8mm] (r+) {$r^L$};
  \node[bddnodex,below right of=r,yshift=-5mm,xshift=-8mm] (r0_4) {$r^R$};
  \node[bddnodex,below left of=p_2,yshift=-5mm,xshift=8mm] (r0_5) {$s^L$};
  \node[bddnodex,below right of=p_2,yshift=-5mm,xshift=-8mm] (r+-) {$s^R$};

  \draw (Rp) coordinate[xshift=-0mm,yshift=-5mm] (Rpa);
  \draw (r) coordinate[xshift=-0mm,yshift=-5mm] (ra);
  \draw (p_2) coordinate[xshift=-0mm,yshift=-5mm] (p_2a);
  \draw[translow] (Lpa)
    to[bend right=15]
    coordinate[pos=0.45] (Lpa_1)
    (q1);

  \draw[trans] (Lp) to 
    (Lpa)
    to[bend left=15]
    coordinate[pos=0.45] (Lpa_2)
    (p_1);

  \filldraw[grayark] (Lpa) to[bend right=10] (Lpa_1) to[bend right=30] (Lpa_2) to[bend right=10] cycle;
  \node at (Lpa) [xshift=-0mm,yshift=-4mm] {$x_1$};

  \draw[translow] (q1a)
    to[bend right]
    coordinate[pos=0.5] (q1a_1)
    (r1);

  \draw[trans] (q1) to 
    node[pos=0.9,left,xshift=-1mm,yshift=1mm] {$\{1\}$}
    (q1a)
    to[bend left]
    coordinate[pos=0.5] (q1a_2)
    (r0_1);

  \filldraw[blueark] (q1a) to[bend right=15] (q1a_1) to[bend
  right=30] (q1a_2) to[bend right=15] cycle;
  \node at (q1a) [xshift=0mm,yshift=-4mm] {$sym$};
  
  \draw[trans] (p_1) to 
    node[pos=0.9,left,xshift=-1mm,yshift=1mm] {$\{1\}$}
    (p_1a)
    to[bend left]
    coordinate[pos=0.5] (p_1a_2)
    (r0_3);

  \draw[translow] (p_1a)
    to[bend right]
    coordinate[pos=0.5] (p_1a_1)
    (r0_2);

  \filldraw[blueark] (p_1a) to[bend right=15] (p_1a_1) to[bend
  right=30] (p_1a_2) to[bend right=15] cycle;
  \node at (p_1a) [xshift=0mm,yshift=-4mm] {$sym$};

  \draw[translow] (Rpa)
    to[bend right=15]
    coordinate[pos=0.45] (Rpa_1)
    (r);

  \draw[trans] (Rp) to 
    (Rpa)
    to[bend left=15]
    coordinate[pos=0.45] (Rpa_2)
    (p_2);

  \filldraw[grayark] (Rpa) to[bend right=10] (Rpa_1) to[bend right=30] (Rpa_2) to[bend right=10] cycle;
  \node at (Rpa) [xshift=-0mm,yshift=-4mm] {$x_1$};

  \draw[translow] (ra)
    to[bend right]
    coordinate[pos=0.5] (ra_1)
    (r+);

  \draw[trans] (r) to 
    node[pos=0.9,left,xshift=-1mm,yshift=1mm] {$\{1\}$}
    (ra)
    to[bend left]
    coordinate[pos=0.5] (ra_2)
    (r0_4);

  \filldraw[blueark] (ra) to[bend right=15] (ra_1) to[bend
  right=30] (ra_2) to[bend right=15] cycle;
  \node at (ra) [xshift=-0mm,yshift=-4mm] {$sym$};

  \draw[translow] (p_2a)
    to[bend right]
    coordinate[pos=0.5] (q1a_1)
    (r0_5);

  \draw[trans] (p_2) to 
    node[pos=0.9,left,xshift=-1mm,yshift=1mm] {$\{1\}$}
    (p_2a)
    to[bend left]
    coordinate[pos=0.5] (q1a_2)
    (r+-);

  \filldraw[blueark] (p_2a) to[bend right=10] (q1a_1) to[bend
  right=30] (q1a_2) to[bend right=15] cycle;
  \node at (p_2a) [xshift=-0mm,yshift=-4mm] {$sym$};

  \draw[link,red] (q1) to[bend left] 
  node[pos=0.7,above,xshift=2mm,yshift=1mm] {$f_i$}
  (r);
  \draw[link,red] (p_1) to[bend left] (p_2);
  \draw[link,red] (r1) to[bend right] 
  node[pos=0.85,below,xshift=2mm,yshift=10mm] {$f'_j$}
  (r+);
  \draw[link,red] (r0_1) to[bend right] (r0_4);
  \draw[link,red] (r0_2) to[bend right] (r0_5);
  \draw[link,red] (r0_3) to[bend right] (r+-);
\end{tikzpicture}
    }\vspace{-2mm}
    \caption{Computing $f'_j$ from $f_i$.}
    \label{fig:inclusionExample}
\end{wrapfigure}
}

\InputIfFileExists{cutmargins.tex}{%
	\pdfpagesattr{/CropBox [125 82 490 688]} % LNCS page made even biger
}{%
}

\Crefname{algocf}{Algorithm}{Algorithms}

\def\orcidID#1{\smash{\href{http://orcid.org/#1}{\protect\raisebox{-1.25pt}{\protect\includegraphics{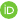}}}}}

\title{\autoqq: From Verification of Quantum Circuits\\ to Verification of Quantum Programs (Technical Report)}
\author{Yu-Fang Chen\inst{1}\orcidID{0000-0003-2872-0336} \and
  Kai-Min Chung\inst{1}\orcidID{0000-0002-3356-369X} \and
  Min-Hsiu Hsieh\inst{2} \and
  Wei-Jia Huang\inst{2}\orcidID{0000-0002-8695-5689} \and \\
  Ond\v{r}ej Leng\'{a}l\inst{3}\orcidID{0000-0002-3038-5875} \and
  Jyun-Ao Lin\inst{1}\orcidID{0000-0001-8560-2147} \and
  Wei-Lun Tsai\inst{1}\orcidID{0009-0003-5832-0867}
}

\institute{Institute of Information Science, Academia Sinica, Taipei, Taiwan
\and
Hon Hai Quantum Computing Research Center, Taipei, Taiwan
\and
Faculty of Information Technology, Brno University of Technology, Brno, Czech Republic
}

\begin{document}

\maketitle

\vspace{-7mm}

\begin{abstract}
We present a~verifier of quantum programs called \autoqq.
Quantum programs extend quantum circuits (the domain of \autoq) by classical
control flow constructs, which enable users to describe advanced quantum
algorithms in a~formal and precise manner.
The extension is highly non-trivial, as we needed to tackle both theoretical
challenges (such as the treatment of measurement, the normalization problem, and
lifting techniques for verification of classical programs with loops to the
quantum world), and engineering issues (such as extending the input format with
a~support for specifying loop invariants).
We have successfully used \autoqq to verify two types of advanced quantum
programs that cannot be expressed using only quantum circuits:
the \emph{repeat-until-success} (RUS) algorithm and the weak-measurement-based
version of Grover's search algorithm.
\autoqq can efficiently verify all our benchmarks: all RUS algorithms were verified instantly and, for the weak-measurement-based version of Grover's search, we were able to handle the case of 100 qubits in $\sim$20 minutes.

% We have developed a verifier called \autoqq, which supports the verification of quantum programs. 
% Quantum programs are an extension of quantum circuits that allow classical control flow. This enables users to describe advanced quantum algorithms in a formal and precise manner. The latest version of $\openqasm$ and other mainstream quantum programming languages already support such extended program constructs, so it is necessary to also support it in our verifier. However, this extension is highly complex. We need to tackle both theoretical challenges (such as the treatment of measurement and the normalization problem, and lifting classical techniques for verification of loop programs to the quantum cases), and engineering issues (such as extending the input format to support specifying loop invariants). We have successfully used \autoqq to verify two types of advanced quantum programs that cannot be expressed using only quantum circuits. These two programs are the \emph{repeat-until-success} (RUS) algorithm~\cite{RUS} and the weakly measurement version~\cite{weakly:chris} of Grover's search algorithm~\cite{Grover96}. Our results demonstrate that \autoqq can efficiently verify all cases. All RUS algorithms were verified instantly, and for the weakly measurement versions of Grover, we were able to handle the case of 100 qubits in around 20 minutes.
\end{abstract}

%%%%%%%%%%%%%%%%%%%%%%%%%%%%%%%%%%%%%%%%%%%%%%%%%%%%%%%%%%%%%%%%%%%%%%%%%%%%%%%%
\vspace{-9.0mm}
\section{Introduction}\label{sec:introduction}
\vspace{-2.0mm}
%%%%%%%%%%%%%%%%%%%%%%%%%%%%%%%%%%%%%%%%%%%%%%%%%%%%%%%%%%%%%%%%%%%%%%%%%%%%%%%%

Quantum \emph{programs} are an extension of quantum \emph{circuits} that
provide users with greater control over quantum computing by allowing them to
use more complex programming constructs like branches and loops.
Some of the most advanced quantum algorithms cannot be defined by quantum circuits alone.
For example, certain class of programs, such as the \emph{repeat-until-success}
(RUS) algorithms~\cite{RUS} (which are commonly used in generating special
quantum gates) and the weak-measurement-based version~\cite{weakly:chris} of
Grover's search algorithm~\cite{Grover96}, use a~loop with the condition
being a~classical value (0~or~1) obtained by measuring a~particular qubit.
This added expressivity presents new challenges, particularly in terms of verification.
The additional complexity comes from the \emph{measurement} operation, where
a~particular qubit is measured to obtain a classical value (and the quantum
state is partially collapsed, which might require \emph{normalization}), and
reasoning about control flow induced by \emph{branches} and
\emph{loops}.

In classical program verification, a~prominent role is played by
\emph{deductive verification}~\cite{Floyd93,Hoare69,HahnleH19}, represented,
e.g., by the tools Dafny~\cite{Leino10}, KeY~\cite{KeY},
Frama-C~\cite{BaudinBBCKKMPPS21}, VeriFast~\cite{JacobsSP10},
VCC~\cite{CohenDHLMSST09}, and many more.
These tools only require the users to provide specifications in the form of
pre- and post-conditions, along with appropriate loop invariants.
The rest of the proving process is entirely (in the ideal case) automated.
Unfortunately, in the realm of quantum computing, similar fully automated
deductive verification tools are, to the best of our knowledge, missing.
Advanced tools for analysis and verification of quantum programs---based on,
e.g., \emph{quantum Hoare logic} and the tool \coqq~\cite{zhou2023coqq} or the
\emph{path-sum} formalism~\cite{amy2018towards} and the
tool~\qbricks~\cite{chareton2021automated}---are quite powerful but require
a~significant amount of human effort.

To bridge this gap, we present \autoqq, a major update over \autoq~\cite{chen2023autoq} with an added support for \emph{quantum programs} (\autoq only supported quantum \emph{circuits}). In \autoq, given a triple $\{P\}\,C\,\{Q\}$, where~$P$ and~$Q$ are the pre- and post-conditions recognizing sets of (pure) quantum states (represented by \emph{tree automata}) and~$C$ is a~quantum circuit, we can verify if all quantum states in~$P$ reach some state in $Q$ after executing $C$. In \autoqq, we addressed several key challenges to make the support of quantum programs possible. First, we need to handle \emph{branch} statements. The key issue here is to handle \emph{measurement} of quantum states whose value is used in a~branch condition.
For this we developed automata-based algorithms
to compute the quantum states after the measurement (\cref{sec:branch}).
The second challenge is the handling of \emph{loop} statements.
Similarly to deductive verification of classical programs, we require the users to provide \emph{an invariant} for each loop.
With the loop invariant provided, we developed a framework handling the rest of
the verification process fully automatically.
Moreover, we show that a naive implementation of the measurement operation will
encounter the \emph{probability amplitude normalization} problem.
This is handled by designing a new algorithm for \emph{entailment testing} (\cref{sec:inclusion}).

Under this framework, the preconditions, postconditions, and invariants are all described using a new automata model called \emph{level-synchronized tree automata (LSTAs)}~\cite{popl}. 
LSTAs are specifically designed to efficiently encode quantum states and gate operations. As the core data structure of the tool, we provide a formal definition of LSTAs in~\cref{sec:lsta} to facilitate the presentation of our new entailment testing approach.

We used \autoqq to verify various quantum programs using the
\emph{repeat-until-success} (RUS) paradigm~\cite{RUS}, as well as the
weak-measurement-based version~\cite{weakly:chris} of Grover's
search~\cite{Grover96} (\cref{sec:experiments}).
\autoqq can efficiently verify all our benchmarks.
The verification process for all RUS algorithms was instantaneous and for the
weakly measured versions of Grover, we were able to handle the case of 100
qubits in $\sim$20\,min.
To the best of our knowledge, \autoqq is currently the only tool for
verification of quantum programs with such a~degree of automation.

\paragraph{Related work.}
Our work aligns with Hoare-style verification of quantum programs, a~topic
extensively explored in prior studies~\cite{zhou2019applied,unruh2019quantum,ying2012floyd,feng2021quantum,liu2019formal}.
This approach, inspired by D’Hondt and Panangaden, utilizes diverse Hermitian
operators as quantum predicates, resulting in a robust and comprehensive proof
system~\cite{d2006quantum}. However, specifying properties with Hermitian
operators is often non-intuitive and difficult for automation due to their vast
matrix sizes. Consequently, these methods are typically implemented using proof
assistants like \coq~\cite{bertot2013interactive},
\isabelle~\cite{wenzel2008isabelle}, or standalone tools built on top of \coq,
like \coqq~\cite{zhou2023coqq}.
These tools require substantial manual effort in the proof search.
The \qbricks approach~\cite{Chareton2021} addresses the challenge of proof search by combining cutting-edge theorem provers with decision procedures, leveraging the Why3 platform~\cite{filliatre2013why3}. %\ja{add coqQ} added one sentence for that...not sure if it is correct need to check 
Nevertheless, this approach still demands considerable human intervention.

In the realm of automatic quantum software analysis tools, 
\emph{circuit equivalence checkers}~\cite{amy2018towards,Coecke2011,hietala2019verified,xu2022quartz,CohenDHLMSST09} prove to be efficient but less flexible in specifying desired properties, 
primarily focusing on equivalence. 
These tools are valuable in compiler validation, with notable examples being~\qcec \cite{burgholzer2020advanced}, \feynman~\cite{amy2018towards}, and \sliqec~\cite{9951285,10.1145/3489517.3530481}.
\emph{Quantum model checking}, supporting a rich specification language (various
temporal logics \cite{FengYY13,MateusRSS09,XuFMD22}), is, due to its limited
scalability, more suited for verifying high-level protocols~\cite{AnticoliPTZ16}.
\qpmc~\cite{FengYY13} stands out as a notable tool in this category. 
Quantum abstract interpretation \cite{yu2021quantum,perdrix2008quantum}
over-approximates the reachable state space to achieve better scalability, but
so far handles only circuits.
The work in~\cite{Diffrule1,Diffrule2} aims at the verification of
\emph{parameterized quantum programs} like \emph{variational quantum eigensolver
(VQE)} or \emph{quantum approximate optimization algorithm (QAOA)}.
However, the correctness properties they focused are very different from what~\autoqq can handle.
While the mentioned tools are fully automated, they serve different purposes or
address different phases of the development cycle compared to \autoqq.

%A formulation of quantum while-loop language was introduced in \cite{YingBook} and extended to the parameterized quantum programs by \cite{Diffrule1,Diffrule2}. 
%The later also developed a framework of \emph{automatic differentiation} on computing quantum programming for bounded and unbounded loops and, unlike us focusing on the partial correctness of programs, their aim is to identify the unknown parameters automatically in the parameterized quantum program design such as variational quantum eigenslover (VQE) or quantum approximate optimization algorithm (QAOA). 

\newcommand{\decisiontreegates}[0]{
\begin{figure}[t]
\begin{minipage}{\linewidth}
\begin{subfigure}[b]{0.2\linewidth}
\centering
\begin{tikzpicture}[anchor=base]
    \node {$x_1$}[sibling distance = 1.2cm, level distance = 0.8cm]
        child  {node {$x_2$} edge from parent [dashed, sibling distance = .7cm]
            child  {node {$a_1$} edge from parent [dashed]}
            child {node {$a_2$} edge from parent [solid]};
            }
        child {node {$x_2$} edge from parent [solid, sibling distance = .7cm]
            child  {node {$a_3$} edge from parent [dashed]}
            child {node {$a_4$}}};
    \end{tikzpicture}
    \caption{State $q$}
    \label{fig:state}
\end{subfigure}
\begin{subfigure}[b]{0.2\linewidth}
\centering
    \begin{tikzpicture}[anchor=base]
    \node {$x_1$}[sibling distance = 1.4cm, level distance = 0.8cm]
        child  {node {$x_2$} edge from parent [dashed, sibling distance = .7cm]
            child  {node {$a_3$} edge from parent [dashed]}
            child {node {$a_4$} edge from parent [solid]};
            }
        child {node {$x_2$} edge from parent [solid, sibling distance = .7cm]
            child  {node {$a_1$} edge from parent [dashed]}
            child {node {$a_2$}}};
    \end{tikzpicture}
    % \vspace{-2mm}
    \caption{Applied $X_1$}
    \label{fig:X1state}
\end{subfigure}
\begin{subfigure}[b]{0.23\linewidth}
\centering
    % \vspace{-2mm}
    \begin{tikzpicture}[anchor=base]
    \node {$x_1$}[sibling distance = 1.4cm, level distance = 0.8cm]
        child  {node {$x_2$} edge from parent [dashed, sibling distance = .7cm]
            child  {node {$a_1$} edge from parent [dashed]}
            child {node {$a_2$} edge from parent [solid]};
            }
        child {node {$x_2$} edge from parent [solid, sibling distance = .7cm]
            child  {node {$a_4$} edge from parent [dashed]}
            child {node {$a_3$}}};
    \end{tikzpicture}
    \caption{Applied $CX^1_2$}
    \label{fig:CX12state}
\end{subfigure}
\begin{subfigure}[b]{0.32\linewidth}
\centering
        \begin{tikzpicture}[anchor=base]
    \node {$x_1$}[sibling distance = 2cm, level distance = 0.8cm]
        child  {node {$x_2$} edge from parent [dashed, sibling distance = 1cm]
            child  {node {$\frac{a_1+a_3}{\sqrt2}$} edge from parent [dashed]}
            child {node {$\frac{a_2+a_4}{\sqrt2}$} edge from parent [solid]};
            }
        child {node {$x_2$} edge from parent [solid, sibling distance = 1cm]
            child  {node {$\frac{a_1-a_3}{\sqrt2}$} edge from parent [dashed]}
            child {node {$\frac{a_2-a_4}{\sqrt2}$}}};
    \end{tikzpicture}

    % \vspace{-2mm}
    \caption{Applied $H_1$}
    \label{fig:H1state}
\end{subfigure}
\end{minipage}
    % \vspace{-4mm}
\caption{The effect of applying selected quantum gates to state $q$.}
    \label{fig:gatestate}
    \vspace{-5mm}
\end{figure}
}

\newcommand{\decisiontreecontrolled}[0]{
\begin{figure}[tb]%\label{fig:decision_tree2}
  % \vspace{-4mm}
\begin{minipage}{\linewidth}
\begin{subfigure}[b]{0.24\linewidth}
\centering
    \begin{tikzpicture}[anchor=base]
    \node {$x_1$}[sibling distance = 1.4cm, level distance = 0.8cm]
        child  {node {$x_2$} edge from parent [dashed, sibling distance = .7cm]
            child  {node {$0$} edge from parent [dashed]}
            child {node {$0$} edge from parent [solid]};
            }
        child {node {$x_2$} edge from parent [solid, sibling distance = .7cm]
            child  {node {$a_0$} edge from parent [dashed]}
            child {node {$a_1$}}};
    \end{tikzpicture}
    % \vspace{-2mm}
    \caption{The initial state}
    \label{fig:initstate}
\end{subfigure}
\hfill
\begin{subfigure}[b]{0.25\linewidth}
\centering
    \begin{tikzpicture}[anchor=base]
    \node {$x_1$}[sibling distance = 1.4cm, level distance = 0.8cm]
        child  {node {$x_2$} edge from parent [dashed, sibling distance = .7cm]
            child  {node {$\frac{a_0}{\sqrt2}$} edge from parent [dashed]}
            child {node {$\frac{a_1}{\sqrt2}$} edge from parent [solid]};
            }
        child {node {$x_2$} edge from parent [solid, sibling distance = .7cm]
            child  {node {$\frac{-a_1}{\sqrt2}$} edge from parent [dashed]}
            child {node {$\frac{-a_0}{\sqrt2}$}}};
    \end{tikzpicture}
    % \vspace{-2mm}
    \caption{Applied $H_1;CX^1_2$}
    \label{fig:H1CX12state}
\end{subfigure}
\hfill
\begin{subfigure}[b]{0.24\linewidth}
\centering
    \begin{tikzpicture}[anchor=base]
    \node {$x_1$}[sibling distance = 1.4cm, level distance = 0.8cm]
        child  {node {$x_2$} edge from parent [dashed, sibling distance = .7cm]
            child  {node {$\frac{a_0}{\sqrt2}$} edge from parent [dashed]}
            child {node {$\frac{a_1}{\sqrt2}$} edge from parent [solid]};
            }
        child {node {$x_2$} edge from parent [solid, sibling distance = .7cm]
            child  {node {$0$} edge from parent [dashed]}
            child {node {$0$}}};
    \end{tikzpicture}
    % \vspace{-2mm}
    \caption{$\measurement{1}{=0}$}
    \label{fig:M0state}
\end{subfigure}
\hfill
\begin{subfigure}[b]{0.24\linewidth}
\centering
    \begin{tikzpicture}[anchor=base]
    \node {$x_1$}[sibling distance = 1.4cm, level distance = 0.8cm]
        child  {node {$x_2$} edge from parent [dashed, sibling distance = .7cm]
            child  {node {$0$} edge from parent [dashed]}
            child {node {$0$} edge from parent [solid]};
            }
        child {node {$x_2$} edge from parent [solid, sibling distance = .7cm]
            child  {node {$\frac{-a_1}{\sqrt2}$} edge from parent [dashed]}
            child {node {$\frac{-a_0}{\sqrt2}$}}};
    \end{tikzpicture}
    % \vspace{-2mm}
    \caption{$\measurement{1}{=1}$}
    \label{fig:M1state}
\end{subfigure}
\end{minipage}
    % \vspace{-4mm}
\caption{Intermediate states during the execution of~\cref{alg:if}.}
\vspace{-5mm}
\end{figure}
}

%%%%%%%%%%%%%%%%%%%%%%%%%%%%%%%%%%%%%%%%%%%%%%%%%%%%%%%%%%%%%%%%%%%%%%%%%%%%%%%%
\vspace{-3.0mm}
\section{Background}\label{sec:background}
\vspace{-2.0mm}
%%%%%%%%%%%%%%%%%%%%%%%%%%%%%%%%%%%%%%%%%%%%%%%%%%%%%%%%%%%%%%%%%%%%%%%%%%%%%%%%
% Before we start, we first provide a minimal quantum and automata theory background that is needed for this work. 
Before we start, we first provide a minimal background needed for this work. 

%*******************************************************************************
\vspace{-3.0mm}
\subsection{The Tree-View of Quantum States}
\vspace{-2.0mm}
%*******************************************************************************

In a traditional computer system with~$n$ bits, a state is represented by~$n$
Boolean values~0 or~1. 
An~$n$-qubit \emph{quantum state} can be viewed as a ``probabilistic distribution'' over $n$-bit classical states. Here we often refer to each classical state as a \emph{computational basis state} or \emph{basis state} for short.
Hence a~quantum state can be represented by a~\emph{binary tree} whose branches correspond to the computational basis states and leaves correspond to \emph{complex probability amplitudes}\footnote{A state with complex amplitude $a+bi$ has the probability $|a+bi|^2=a^2+b^2$ of being observed.}.

In \cref{fig:state}, we can see an example of a~2-qubit state~$q$ that maps basis
states $\ket{00}$, $\ket{01}$, $\ket{10}$, $\ket{11}$ to probability amplitudes
$a_1$, $a_2$, $a_3$, and $a_4$, respectively.
The left-going dashed line denotes the $0$-branch, and the right-going solid
line denotes the $1$-branch.
A~quantum state can also be represented as a~formal sum of basis states
multiplied by their amplitudes, e.g.,
we can represent the state~$q$ as
$a_1\ket{00}+a_2\ket{01}+a_3\ket{10}+a_4\ket{11}$.

\decisiontreegates

\emph{Quantum gates} are fundamental operations performed on quantum states.
Basic quantum gates and their effects on the state $q$ from \cref{fig:state}
are shown in \cref{fig:X1state,fig:H1state,fig:CX12state}.
To specify the target qubit to which a single qubit gate $U$ is applied,
a~subscript number~$i$ is used.
For example, $X_i$ denotes the application of the~$X$ gate, which is also known
as the quantum ``negation'' gate, to the $i$-th qubit.
The effect of this gate is to swap the $0$- and $1$-subtrees under all $x_i$
nodes (cf. \cref{fig:X1state}).
On the other hand, for a~controlled gate, a~superscript number~$i$ is used to
indicate the control qubit, while a subscript number~$j$ is used for the target
qubit.
The most notable example is the controlled-$X$ gate $CX^i_j$, which applies the
$X_j$ gate to $1$-subtrees under (assuming $i<j$) all $x_i$ nodes (cf.
\cref{fig:CX12state}).
% The result of applying $X_1$, $H_1$, and $CX^1_2$ to $q$ can be found in~\cref{fig:X1state},~\cref{fig:H1state}, and~\cref{fig:CX12state}, respectively.
Observe that after applying an~$H$ gate (the \emph{Hadamard} gate, which creates
a~quantum superposition; cf. \cref{fig:H1state}), there is a~\emph{normalization factor}
$\frac{1}{\sqrt{2}}$ on all leaves to keep the sum of probabilities one.
This normalization factor can be derived from the tree leaf values $(a_1+a_3)$,
$(a_2+a_4)$, $(a_1-a_3)$, and $(a_2-a_4)$ and the fact that $\sum_{i=1}^4|a_i|^2=1$.

%*******************************************************************************
% \vspace{-3.0mm}
\subsection{Level-Synchronized Tree Automata}\label{sec:lsta}
% \vspace{-2.0mm}
%*******************************************************************************
As we mentioned in~\cref{sec:introduction}, the new algorithms introduced in
this work are built on top of LSTAs, making it essential to provide a~formal definition.
Readers may choose to skim this section initially and refer back to it for
details as needed later.

\subsubsection{Trees.}
Our framework is based on the concept of perfect binary trees.
% Before introducing the mathematical model used in our framework, we have to define a perfect binary tree formally first.
A~\emph{perfect binary tree} $T$ is a map from \emph{tree nodes}
% $\displaystyle\bigcup_{0\le i\le n}\{0,1\}^i$
$\bigcup_{0\le i\le n}\{0,1\}^i$, for some $n\in\natz\coloneq\nat\cup\{0\}$,
to a~nonempty set of symbols~$\Sigma$, i.e., $T\colon\bigcup_{0\le i\le
n}\{0,1\}^i \to \Sigma$.
All nodes
% $v \in \displaystyle\bigcup_{0\le i < n}\{0,1\}^i$
$v \in \bigcup_{0\le i < n}\{0,1\}^i$
are \emph{internal} and have children nodes $v.0$ (left) and $v.1$ (right)
where `$.$' denotes concatenation (we denote the empty string by~$\epsilon$).
All nodes $v \in \{0,1\}^n$ are \emph{leaves} and have no children.
A~node $v$'s \emph{height} is its word length, denoted $|v|$. A node $v$ is at tree level $i$ when $|v|=i$. We denote $T$'s height by $n$.
Perfect binary trees can be used to represent quantum states or vectors of the size~$2^n$.
For instance, the quantum state of~\cref{fig:state} corresponds to a perfect binary tree %$T$ with
% $T(\epsilon)=x_1$, $T(0)=T(1)=x_2$, $T(00)=a_1$, $T(01)=a_2$, $T(10)=a_3$, and $T(11)=a_4$,
$T = \{ \epsilon \mapsto x_1,0 \mapsto x_2, 1 \mapsto x_2, 00 \mapsto a_1, 01 \mapsto a_2, 10 \mapsto a_3, 11 \mapsto a_4\}$ of height~2.
Children of the node $1$ are $10$ and $11$, and the leaf node $10$ has no children.

\subsubsection{LSTAs.} A~\emph{(symbolic) level-synchronized tree automaton (\lsta)}
\cite{popl} is a tuple $\aut = \tuple{Q, \mathbb N \cup \symterm, \Delta,
\rootstates, \globconstr}$ where $Q$ is a finite set of \emph{states},
$\rootstates \subseteq Q$ is a set of \emph{root states}, $\symterm$ is a set
of terms obtained from complex numbers $\complex$ and a set of \emph{complex
variables} $\vars$ using function symbols from some fixed theory
(in this paper, we will use $\nat$ for internal node symbols and $\complex \cup
\symterm$ for leaf symbols).
$\Delta$ is a~set of \emph{transitions} of the form $\delta_i = \ctranstreenoset
q f {q_1,q_2}{C}$ (\emph{internal transitions}) and $\delta_\ell =
\ctranstreenoset q f {}{C}$ (\emph{leaf transitions}), where $q, q_1, q_2 \in
Q$, $f\in\symterm$, and $C\subseteq \nat$ is a~finite set of \emph{choices}.
In the rest of the paper, we also draw the internal transition $\delta_i$ and
leaf transition $\delta_\ell$ as $\tikztrans{q}{C}{f}{q_1}{q_2}$ and
$\tikzleaftrans{q}{C}{f}$, respectively, to provide a more intuitive presentation.
In the aforementioned form, we call $q$, $f$, $C$, $q_1$, $q_2$, and
$\{q_1,q_2\}$ the \emph{top}, \emph{symbol}, \emph{choices}, \emph{left},
\emph{right}, and \emph{bottom}, respectively, of the transition $ \delta_i$,
and denote them by $\topof { \delta_i}$, $\symof { \delta_i}$, $\ell(
\delta_i)$, $\leftof{\delta_i}$, $\rightof{\delta_i}$, and $\botof{\delta_i}$,
respectively.
Needless to say, $\botof{ \delta_\ell} = \emptyset$.
We further require the choices of different transitions with the same top state
to be disjoint, i.e., $\forall \delta_1\ne\delta_2\in\Delta\colon
\topof{\delta_1}=\topof{\delta_2}\implies\ell(\delta_1)\cap\ell(\delta_2)=\emptyset$.
Finally, the \emph{global constraint} $\globconstr$ is a formula used to
restrict the values of $\mathbb{X}$ to those that satisfy~$\globconstr$ (if not
stated, it is assumed to be ``true'').
For instance, when $\mathbb{X} = \{a,b\}$, we can set $\globconstr = |a|>|b|$
to restrict the allowed valuation of $a$ and $b$.

\subsubsection{Semantics of LSTAs.}
A \emph{run} of an \lsta $\aut$ on a~tree $T$ is a map
$\run\colon \dom(T) \rightarrow \Delta$ from tree nodes to transitions of
$\aut$ such that for each node $v\in \dom(T)$, when $v$ is an internal node,
$\run(v)$ is of the form $\ctranstreenoset {q} {T(v)} {\topof{\run(v.0)},\
\topof{\run(v.1)}} C$.
When $v$ is a~leaf node, $\run(v)$ is of the form $\ctranstreenoset {q} {T(v)} {} C$.
%For the sake of discussion, we also define a~run on a~\emph{partial} tree $T'$ \ol{what is this?} as a~map $\run'\colon \dom(T') \rightarrow \Delta$ except that when$v\in\dom(T')$ is a leaf node, $\rho'(v)$ may be an internal transition.
% The intuition behind this is that we try to fit the tree $T$ with a set of
% transitions of $\aut$ by assigning to each node of $T$ a~transition of $\aut$
% such that each node is assigned a unique state in~$Q$.
We give a \emph{run} $\rho$ of the \lsta $\aut$ in \cref{fig:algifpost} on the
tree $T$ in \cref{fig:M1state} (used later in \cref{alg:if}) as follows:
\begin{flalign*}
  &&
  \rho(\epsilon)&{}= \tikztrans {p} {\{1\}} {x_1} {q_0}{q_-}, 
  &
  \rho(0)&{} =\tikztrans {q_{0}}{\{1\}} {x_2} {r_{0}}{r_{0}},\\
  && 
  \rho(1)&{}=\tikztrans {q_{-}}{\{1\}} {x_2} {r_{3}}{r_{4}},
  &
  \rho(00) = \rho(01)&{}=\tikzleaftrans {r_{0}}{\{1\}} {0},\\
  &&
  \rho(10)&{}=\tikzleaftrans {r_{3}}{\{1\}} {-\tfrac{a_1}{\sqrt{2}}},
  &
  \rho(11)&{}=\tikzleaftrans {r_{4}}{\{1\}} {-\tfrac{a_0}{\sqrt{2}}}.
 \hspace{-14mm}
\end{flalign*}
%
% \begin{flalign*}
%   &&
%   \rho(\epsilon)&{}= \ctranstree {p} {x_1} {q_0, q_-}{1}, 
%   &
%   \rho(0)&{} =\ctranstree {q_{0}} {x_2} {r_{0}, r_{0}}{1},\\
%   && 
%   \rho(1)&{}=\ctranstree {q_{-}} {x_2} {r_{3}, r_{4}}{1},
%   &
%   \rho(00) = \rho(01)&{}=\ctranstree {r_{0}} {0} {}{1},\\
%   &&
%   \rho(10)&{}=\ctranstree {r_{3}} {-\tfrac{a_1}{\sqrt{2}}} {}{1},
%   &
%   \rho(11)&{}=\ctranstree {r_{4}} {-\tfrac{a_0}{\sqrt{2}}} {}{1}.
%  \hspace{-14mm}
% \end{flalign*}
%
% We give a \emph{run} $\rho$ of the \lsta $\aut$ in \cref{fig:algifpost} on the
% tree $T$ in \cref{fig:M0state} (used later in \cref{alg:if}) as follows:
% \begin{flalign*}
%   &&
%   \rho(\epsilon)&{}= \ctranstree {p} {x_1} {q_+, q_0}{1}, 
%   &
%   \rho(0)&{} =\ctranstree {q_{+}} {x_2} {r_{1}, r_{2}}{1},\\
%   && 
%   \rho(1)&{}=\ctranstree {q_{0}} {x_2} {r_{0}, r_{0}}{1},
%   &
%   \rho(00)&{}=\ctranstree {r_{1}} {\tfrac{a_0}{\sqrt{2}}} {}{1},\\
%   &&
%   \rho(01)&{}=\ctranstree {r_{2}} {\tfrac{a_1}{\sqrt{2}}} {}{1},
%   &
%   \rho(10)=\rho(11)&{}=\ctranstree {r_0} {0} {}{1}.& 
%  \hspace{-14mm}
% \end{flalign*}
%
%\linelabel{ln:acceptingrun}
A run $\rho$ is \emph{accepting} if $\topof{\rho(\epsilon)}\in\rootstates$ and
all transitions used at the same level share some common choice, i.e.,
% $\displaystyle\forall\ 0\le i\le n\colon\bigcap_{\delta\in\{\rho(v)\ |\ v\in\{0,1\}^i\}}\ell(\delta)\neq\emptyset$.
$\forall 0\le i\le n\colon\bigcap\{\ell(\delta) \mid \delta\in\{\rho(v) \mid v\in\{0,1\}^i\}\}\neq\emptyset$
(this is the essential difference from standard tree automata that gives LSTAs
the power to compactly represent some classes of quantum states).
The \emph{language} of~$\aut$, denoted as $\langof \aut$ is then a~set of trees~$T$
over $\mathbb N \cup \symterm$ such that there exists an accepting run
of~$\aut$ over~$T$.
Given a~tree~$T$ over $\nat \cup \complex$ and an assignment $\sigma\colon \vars \to \complex$, we use~$T[\sigma]$ to denote the tree obtained from~$T$ by
\begin{inparaenum}[(i)]
  \item  substituting all occurrence of variables~$x \in \vars$ in~$T$ by $\sigma(x)$ and
  \item  evaluating all terms in the resulting tree into values $c \in \complex$.
\end{inparaenum}
%
% Finally, the \emph{semantics} of~$\aut$, denoted as~$\semof \aut$, is the set
% of trees~$T_\complex$ over~$\nat \cup \complex$ such that there exists a~tree
% $T \in \langof \aut$ and an assignment $\sigma\colon \vars \to \complex$ that
% is a~model of~$\globconstr$ for which it holds that $T_\complex = T[\sigma]$.
% Formally $\semof \aut = \{T[\sigma] \mid T \in \langof \aut, \sigma \models \globconstr\}$.

% A run represents exactly one tree, but there may be many different runs representing the same tree.
% This concept plays an essential role in filtering out those unwanted generated trees.
%---in other words, transitions at each tree level are \emph{synchronized}.

\decisiontreecontrolled

\vspace{-3.0mm}
\section{Overview}\label{sec:overview}
\vspace{-2.0mm}
%%%%%%%%%%%%%%%%%%%%%%%%%%%%%%%%%%%%%%%%%%%%%%%%%%%%%%%%%%%%%%%%%%%%%%%%%%%%%%%%

\newcommand{\figIf}[0]{
\begin{wrapfigure}[7]{r}{54mm}
\vspace*{-8mm}
\hspace*{-2mm}
\scalebox{0.9}{
\begin{minipage}{59mm}
\begin{algorithm}[H]
  \caption{``$-X_2$'' if $M_1 = 1$}
  \label{alg:if}
  {\color{gray} Pre: $\{\specpair{10}{a_0}+ \specpair{11}{a_1}\}$\;}
  $H_1;CX^1_2$\;
  \lIf{$M_1=0$}{\{$X_1$\}\label{ln:ifmeasure}}
  {\color{gray} Post: $\{\specpair{10}{a_0}+ \specpair{11}{a_1},$}\\
  \hspace*{7.1mm}\color{gray}{$\specpair{10}{{-}a_1}- \specpair{11}{a_0})\}$\;}
\end{algorithm}
\end{minipage}
}
\end{wrapfigure}
}

\newcommand{\figWhile}[0]{
\begin{wrapfigure}[7]{r}{68mm}
\vspace*{-8mm}
\hspace*{-3mm}
\scalebox{0.9}{
\begin{minipage}{76mm}
\begin{algorithm}[H]
  % \caption{Apply $-X$ gate on the second qubit, loop statement version}
  \caption{``$-X_2$''}
  \label{alg:loop}
{\color{gray} Pre: $\{\specpair{10}{a_0}+ \specpair{11}{a_1}\}$\;}
%  $q := \ket{1}\bigotimes \ket{\phi}$\;
  $H_1;CX^1_2$\;
{\color{gray} Inv: $\{\specpair{00}{\frac{a_0}{\sqrt{2}}}+
  \specpair{01}{\frac{a_1}{\sqrt{2}}}- \specpair{10}{\frac{a_1}{\sqrt{2}}}-
  \specpair{11}{\frac{a_0}{\sqrt{2}}}\}$\;}
  \lWhile{$M_1=0$}{\{$X_1;H_1;CX^1_2$\}}
{\color{gray} Post: $\{\specpair{10}{\mathrm{-}a_1}- \specpair{11}{a_0}\}$\;}
\end{algorithm}
\end{minipage}
}
\end{wrapfigure}
}

\figIf
In this section, we provide an overview of automata-based quantum program
verification with a~running example (chosen for its simplicity).
In the example, the quantum program creates the effect of a non-standard quantum
gate ``$-X$'' (applying the $X$ gate and negating all amplitude values) using
the standard gates $X$, $H$, and $CX$ (\cref{alg:if}). The program operates on
a~2-qubit system and performs the ``$-X$'' gate on the second qubit when the first qubit is measured to be~1; otherwise, (the first qubit
is~0 after measurement) the state stays unchanged.

For all $a_0,\ a_1\in\mathbb C$, we verify \cref{alg:if} against the precondition $\{\specpair{10}{a_0}+
\specpair{11}{a_1}\}$, which allows the state with the first qubit
$\ket{1}$ and the second qubit $a_0\ket{0}+a_1\ket{1}$, and the postcondition
$\{\specpair{10}{a_0}+ \specpair{11}{a_1},\ \specpair{10}{{-}a_1}- \specpair{11}{a_0}\}$,
which includes the original state and the state after the ``$-X_2$'' gate.
%Here, $a_0$ and $a_1$ are complex variables.

Our approach first constructs two LSTAs $P$ and $Q$ that can recognize the
states (binary trees) of the pre- and post-conditions, respectively, and then computes another \lsta by executing the gates $H_1;CX^1_2$ from $P$ (see 
\cref{fig:initstate} for the only quantum state accepted by $P$) with the gate algorithm introduced in~\cite{popl}.
This results in an LSTA $P_1$ that recognizes the state shown in~\cref{fig:H1CX12state}. After applying the operator $M_1$ to measure $x_1$ (\lnref{ln:ifmeasure}), $P_1$ splits into two copies. One copy, $P_2$, accepts the only quantum state shown in \cref{fig:M0state}, where the first qubit is measured to be $0$. The other copy, $P_3$, accepts the only quantum state shown in \cref{fig:M1state}, where the first qubit is measured to be $1$.

It is important to note that the probability amplitudes of the quantum states from
\cref{fig:M0state,fig:M1state} have not been normalized yet.
To do that, we need to multiply all leaves with the normalization factor~$\sqrt{2}$.
This will ensure that the square sum of their absolute amplitude values
becomes~1.
Although the quantum states are not yet normalized, we, however, still have sufficient
information to obtain the corresponding normalized states.
In~\autoqq, we choose to ignore all normalization factors and design a new
entailment testing algorithm~(\cref{sec:inclusion}) that can detect the
equivalence of two non-normalized states.
After both the true and false paths of the \textbf{if} statement in the example are
processed, we obtain two LSTAs $P_2$ and $P_3$ capturing all reachable states.
We then construct their union and test if all states in the union are included
in the post-condition (recognized by $Q$) by testing entailment.

\figWhile
A~drawback of \cref{alg:if} is that the desired effect ``$-X$'' manifests only
if $M_1 = 1$.
In the case $M_1 = 0$, we need to run the same
algorithm again until we get a measurement of 1.
To achieve this, we can use a~\textbf{while} loop statement, as shown
in~\cref{alg:loop}.
The loop allows us to repeatedly execute the same branch statement
until the desired outcome is achieved.

To verify that the loop works correctly, we require the user to provide a loop
invariant in the form of an LSTA.
The invariant here is $\{\specpair{00}{\frac{a_0}{\sqrt{2}}}+
  \specpair{01}{\frac{a_1}{\sqrt{2}}}- \specpair{10}{\frac{a_1}{\sqrt{2}}}-
  \specpair{11}{\frac{a_0}{\sqrt{2}}}\}$ (cf.~\cref{fig:H1CX12state}).
The verification process then involves checking if the invariant is
\emph{inductive}, covers all reachable states before entering the loop, and does
not contain any state that would violate the post-condition. More details on the verification process will be given in \cref{sec:loop}. With the loop invariant provided, we can ensure that the algorithm ends up with a~state where the ``$-X$'' gate is applied to the second qubit when it terminates.

 \begin{wrapfigure}[12]{r}{0.4\textwidth}
\centering
\vspace{-6mm}
\resizebox{\linewidth}{!}{
  \begin{tikzpicture}[>=stealth',node distance=20mm,rotate=37]
  \pgfsetlinewidth{1bp}
  \tikzstyle{bddnode}=[draw,rectangle,rounded corners=2mm]
  \tikzstyle{bddleaf}=[]
  \tikzstyle{trans}=[->,>=stealth']
  \tikzstyle{translow}=[->,>=stealth',dashed]
  \tikzstyle{stick}=[-,>=stealth']
  % \tikzstyle{stick}=[->,>=stealth']
  \tikzstyle{hidtrans}=[]
  \tikzstyle{ark}=[]
  \tikzstyle{blueark}=[fill=blue,opacity=0.2]
  \tikzstyle{redark}=[fill=red,opacity=0.6]

  \tikzstyle{outp}=[scale=0.75,fill=black!30,inner sep=0.6mm]

  \tikzstyle{bddnodex}=[bddnode,inner sep=1mm]

  % NODES
  \node[bddnodex] (p) {$p$};
  \node[right of=p,xshift=-10mm] (root) {};
  \node[bddnodex,below left of=p,yshift=-5mm] (q+) {$q_+$};
  \node[bddnodex,below of=p,yshift=0.5mm] (q0) {$q_0$};
  \node[bddnodex,below right of=p,yshift=-5mm] (q-) {$q_-$};
  \node[bddnodex,below left of=q+,yshift=-5mm] (r1) {$r_1$};
  \node[bddnodex,below of=q+,yshift=0mm] (r2) {$r_2$};
  \node[bddnodex,below right of=q-,yshift=-5mm] (r4) {$r_4$};
  \node[bddnodex,below of=q-,yshift=0mm] (r3) {$r_3$};
  \node[bddnodex,below right of=q+,yshift=-5mm] (r0) {$r_0$};
  % \node[bddnodex,below right of=q0,yshift=-5mm] (r1-) {$r_\pm$};
  \node[bddleaf, below of=r1,yshift=11mm] (r1a0) {$\frac{a_0}{\sqrt{2}}$};
  \node[bddleaf, below of=r2,yshift=11mm] (r2a0) {$\frac{a_1}{\sqrt{2}}$};
  \node[bddleaf, below of=r0,yshift=11mm] (r0a) {$0$};
  \node[bddleaf, below of=r3,yshift=11mm] (r3a0) {$\frac{-a_1}{\sqrt{2}}$};
  \node[bddleaf, below of=r4,yshift=11mm] (r4a0) {$\frac{-a_0}{\sqrt{2}}$};

  \draw (p) coordinate[xshift=-5mm,yshift=-5mm] (pa);
  \draw (p) coordinate[xshift= 5mm,yshift=-5mm] (pb);

  \draw (q+) coordinate[xshift=0mm,yshift=-6mm] (q+a);
  \draw (q-) coordinate[xshift=0mm,yshift=-6mm] (q-a);
  % \draw (q+) coordinate[xshift= 5mm,yshift=-5mm] (q+b);

  \draw (q0) coordinate[xshift=0mm,yshift=-6mm] (q0a);
  % \draw (q0) coordinate[xshift= 5mm,yshift=-5mm] (q0b);

  % \draw (r0) coordinate[xshift=-0mm,yshift=-5mm] (r0a);

  %%%%%%%%%%%%%%% transition %%%%%%%%%%%%%%%%%%%
  \draw[trans,-] (p)
    to
    (pa);
  \draw[trans] (pa)
    to[bend right=10]
    node[pos=0,left,xshift=-1mm,yshift=1mm] {$\{1\}$}
    coordinate[pos=0.45] (pa1)
    (q+);
  \draw[translow] (pa)
    % node[pos=0.9,left,xshift=-1mm,yshift=1mm] {$\{1\}$}
    % (pa)
    to[bend left=10]
    coordinate[pos=0.45] (pa2)
    (q0);
  \draw[translow] %(pb) to 
    % node[pos=0.9,right,xshift=+1mm,yshift=1mm] {}
    (pb)
    to[bend right=10]
    coordinate[pos=0.45] (pb1)
    (q0);
  \draw[trans] (p) to 
    node[pos=0.9,right,xshift=+1mm,yshift=1mm] {$\{2\}$}
    (pb)
    to[bend left=10]
    coordinate[pos=0.45] (pb2)
    (q-);

  \filldraw[blueark] (pa) to[bend right=5] (pa1) to[bend right=50] (pa2) to[bend right=5] cycle;
  \filldraw[blueark] (pb) to[bend right=5] (pb1) to[bend right=50] (pb2) to[bend right=5] cycle;
  \node at (pa) [xshift=-0mm,yshift=-4mm] {$x_1$};
  \node at (pb) [xshift=+1mm,yshift=-4mm] {$x_1$};

  %%%%%%%%%%%%%%% transition %%%%%%%%%%%%%%%%%%%
  \draw[translow] (q+a)
    to[bend right]
    coordinate[pos=0.5] (q+a1)
    (r1);

  \draw[trans] (q+) to 
    node[pos=0.9,right,xshift=0mm,yshift=1mm] {$\{1\}$}
    (q+a)
    to%[bend right]
    coordinate[pos=0.6] (q+a2)
    (r2);

  \filldraw[blueark] (q+a) to[bend right=15] (q+a1) to[bend
  right=30] (q+a2) to cycle;
  \node at (q+a) [xshift=-3mm,yshift=-3mm] {$x_2$};
  \node at (q-a) [xshift=+3mm,yshift=-3mm] {$x_2$};

  \draw[translow] (q-a) to 
    node[pos=0.9,left,xshift=0mm,yshift=1mm] {$\{1\}$}
    (q-a)
    to%[bend right]
    coordinate[pos=0.6] (q-a1)
    (r3);

  \draw[trans] (q-) to 
    % node[pos=0.9,right,xshift=0mm,yshift=1mm] {$\{1\}$}
    (q-a)
    to[bend left]%=20]
    coordinate[pos=0.5] (q-a2)
    (r4);
  \filldraw[blueark] (q-a) to[bend left=15] (q-a2) to[bend
  left=30] (q-a1) to cycle;
  
  %%%%%%%%%%%%%%% transition %%%%%%%%%%%%%%%%%%%
  % \draw[trans] (q+) to 
  %   node[pos=0.9,right,xshift=0mm,yshift=2mm] {$\{2\}$}
  %   (q+b)
  %   to[bend left]
  %   coordinate[pos=0.3] (q+b1)
  %   (r1);

  % \draw[translow] (q+b)
  %   to[bend left]
  %   coordinate[pos=0.6] (q+b2)
  %   (r0);

  % \filldraw[blueark] (q+b) to[bend left=10] (q+b1) to[bend
  % right=30] (q+b2) to[bend right=15] cycle;
  % \node at (q+b) [xshift=2mm,yshift=-4mm] {$x_2$};

  %%%%%%%%%%%%%%% transition %%%%%%%%%%%%%%%%%%%
  \draw[translow] (q0a)
    to[bend right]
    coordinate[pos=0.6] (q1a1)
    (r0);

  \draw[trans] (q0) to 
    node[pos=0.9,left,xshift=0mm,yshift=1mm] {$\{1\}$}
    (q0a)
    to[bend left]
    coordinate[pos=0.6] (q1a2)
    (r0);

  \filldraw[blueark] (q0a) to[bend right=15] (q1a1) to[bend right=30] (q1a2) to[bend right=15] cycle;
  \node at (q0a) [xshift=0mm,yshift=-4mm] {$x_2$};

  %%%%%%%%%%%%%%% transition %%%%%%%%%%%%%%%%%%%
  % \draw[trans] (q0) to 
  %   node[pos=0.9,right,xshift=0mm,yshift=2mm] {$\{2\}$}
  %   (q0b)
  %   to[bend left]
  %   coordinate[pos=0.3] (q1b1)
  %   (r0);

  % \draw[translow] (q0b)
  %   to[bend left]
  %   coordinate[pos=0.6] (q1b2)
  %   (r1-);

  % \filldraw[blueark] (q0b) to[bend left=10] (q1b1) to[bend
  % right=30] (q1b2) to[bend right=15] cycle;
  % \node at (q0b) [xshift=2mm,yshift=-4mm] {$x_2$};

  %%%%%%%%%%%%%%% transition %%%%%%%%%%%%%%%%%%%
  \draw[trans] (root) to (p);
  \draw[stick] (r1) to node[left,xshift=-1mm] {$\{1\}$} (r1a0);
  \draw[stick] (r2) to node[left,xshift=-1mm] {$\{1\}$} (r2a0);
  \draw[stick] (r0) to node[left, xshift=-1mm,yshift=0mm] {$\{1\}$} (r0a);
  \draw[stick] (r3) to node[left,xshift=-1mm] {$\{1\}$} (r3a0);
  \draw[stick] (r4) to node[left,xshift=-1mm] {$\{1\}$} (r4a0);
  
\end{tikzpicture}
}
\vspace{-6mm}
 \caption{The \lsta recognizing the postcondition of~\cref{alg:if}.}
 \label{fig:algifpost}
 \end{wrapfigure} 
  %%%%%%%%%%%%%%%%%%%%%%%%%%%%%%%%%%%%%%%%%%%%

In \autoqq, preconditions, postconditions, and invariants are represented as
sets of quantum states, encoded using the LSTA model.
Therefore, it is important for users to understand how to encode a set of
quantum states with an LSTA.
Below, we provide two examples to give a basic understanding of the process.
In the first example, we show how to encode the postcondition of~\cref{alg:if}, $\{\specpair{10}{a_0}+ \specpair{11}{a_1},\ \specpair{10}{\mathrm{-}a_1}- \specpair{11}{a_0}\}$. 

The corresponding LSTA is shown in~\cref{fig:algifpost}. The LSTA constructs
trees that depict quantum states beginning from the initial state $p$ at the
root. It continues to build the tree by choosing transitions to explore new
child states at each step, and this process continues until it reaches the
leaves. For instance, the tree in~\cref{fig:M0state} can be generated by first
picking the transition \tikztrans{p}{\{1\}}{x_1}{q_{+}}{q_0}, then the two
transitions
\tikztrans{q_+}{\{1\}}{x_2}{r_1}{r_2} and \tikztransone{q_0}{\{1\}}{x_2}{r_0}\ ,
and ending with the three leaf transitions
\tikzleaftrans{r_1}{\{1\}}{\nicefrac{a_0}{\sqrt{2}}},
\tikzleaftrans{r_2}{\{1\}}{\nicefrac{a_1}{\sqrt{2}}}, and
\tikzleaftrans{r_0}{\{1\}}{0}.
Similar to the conventional tree automata (TAs)
model, LSTAs utilize disjunctive branches to represent various states that
share a common structure. In~\cref{fig:algifpost}, the state $p$ has two
possible outgoing transitions. If one picks the other transition
\tikztrans{p}{\{2\}}{x_1}{q_0}{q_-} at the beginning, we can generate the
tree shown in~\cref{fig:M1state}.

\begin{figure}[t]\centering
\begin{subfigure}{0.6\textwidth}
\scalebox{.8}{  
  \begin{tikzpicture}[>=stealth',node distance=20mm,rotate=37]

  \pgfsetlinewidth{1bp}
  \tikzstyle{bddnode}=[draw,rectangle,rounded corners=2mm]
  \tikzstyle{bddleaf}=[]
  \tikzstyle{trans}=[->,>=stealth']
  \tikzstyle{translow}=[->,>=stealth',dashed]
  \tikzstyle{stick}=[-,>=stealth']
  % \tikzstyle{stick}=[->,>=stealth']
  \tikzstyle{hidtrans}=[]
  \tikzstyle{ark}=[]
  \tikzstyle{blueark}=[fill=blue,opacity=0.2]
  \tikzstyle{redark}=[fill=red,opacity=0.6]

  \tikzstyle{outp}=[scale=0.75,fill=black!30,inner sep=0.6mm]

  \tikzstyle{bddnodex}=[bddnode,inner sep=1mm]

  % NODES

  \node[bddnodex] (p) {$p$};
  \node[right of=p,xshift=-10mm] (root) {};
  \node[bddnodex,below left of=p,yshift=-5mm] (q+) {$q_+$};
  \node[bddnodex,below right of=p,yshift=-5mm] (q+-) {$q_\pm$};
  \node[bddnodex,below left of=q+,yshift=-5mm] (r+) {$r_+$};
  \node[bddnodex,below right of=q+,yshift=-5mm] (r0) {$r_0$};
  \node[bddnodex,below right of=q+-,yshift=-5mm] (r+-) {$r_\pm$};
  \node[bddleaf, below of=r+,yshift=11mm] (r+a) {$\frac{1}{\sqrt{2}}$};
  \node[bddleaf, below of=r0,yshift=11mm] (r0a) {$0$};
  \node[bddleaf, below of=r+-,xshift=-5mm,yshift=11mm] (r+-a) {$\frac{1}{\sqrt{2}}$};
  \node[bddleaf, below of=r+-,xshift= 5mm,yshift=11mm] (r+-b) {$\frac{-1}{\sqrt{2}}$};

  \draw (p) coordinate[xshift=-0mm,yshift=-5mm] (pa);
  \draw (p) coordinate[xshift= 5mm,yshift=-5mm] (pb);

  \draw (q+) coordinate[xshift=-5mm,yshift=-5mm] (q+a);
  \draw (q+) coordinate[xshift= 5mm,yshift=-5mm] (q+b);

  \draw (q+-) coordinate[xshift=-5mm,yshift=-5mm] (q+-a);
  \draw (q+-) coordinate[xshift= 5mm,yshift=-5mm] (q+-b);

  % \draw (r0) coordinate[xshift=-0mm,yshift=-5mm] (r0a);

  %%%%%%%%%%%%%%% transition %%%%%%%%%%%%%%%%%%%
  \draw[translow] (pa)
    to[bend right=10]
    coordinate[pos=0.45] (pa1)
    (q+);

  \draw[trans] (p) to 
    node[pos=0.9,left,xshift=-1mm,yshift=1mm] {$\{1\}$}
    (pa)
    to[bend left=10]
    coordinate[pos=0.45] (pa2)
    (q+-);

  \filldraw[blueark] (pa)
    to[bend right=5] (pa1)
    to[bend right=50] (pa2)
    to[bend right=5] cycle;
  \node at (pa) [xshift=-0mm,yshift=-4mm] {$x_1$};

  %%%%%%%%%%%%%%% transition %%%%%%%%%%%%%%%%%%%
  \draw[translow] (q+a)
    to[bend right]
    coordinate[pos=0.6] (q+a1)
    (r+);

  \draw[trans] (q+) to 
    node[pos=0.9,left,xshift=0mm,yshift=2mm] {$\{1\}$}
    (q+a)
    to[bend right]
    coordinate[pos=0.3] (q+a2)
    (r0);

  \filldraw[blueark] (q+a) to[bend right=15] (q+a1) to[bend
  right=30] (q+a2) to[bend left=10] cycle;
  \node at (q+a) [xshift=-2mm,yshift=-4mm] {$x_2$};

  %%%%%%%%%%%%%%% transition %%%%%%%%%%%%%%%%%%%
  \draw[trans] (q+) to 
    node[pos=0.9,right,xshift=0mm,yshift=2mm] {$\{2\}$}
    (q+b)
    to[bend left]
    coordinate[pos=0.3] (q+b1)
    (r+);

  \draw[translow] (q+b)
    to[bend left]
    coordinate[pos=0.6] (q+b2)
    (r0);

  \filldraw[blueark] (q+b)
    to[bend left=10] (q+b1)
    to[bend right=30] (q+b2)
    to[bend right=15] cycle;
  \node at (q+b) [xshift=2mm,yshift=-4mm] {$x_2$};

  %%%%%%%%%%%%%%% transition %%%%%%%%%%%%%%%%%%%
  \draw[translow] (q+-a)
    to[bend right]
    coordinate[pos=0.6] (q1a1)
    (r0);

  \draw[trans] (q+-) to 
    node[pos=0.9,left,xshift=0mm,yshift=2mm] {$\{1\}$}
    (q+-a)
    to[bend right]
    coordinate[pos=0.3] (q1a2)
    (r+-);

  \filldraw[blueark] (q+-a)
    to[bend right=15] (q1a1)
    to[bend right=30] (q1a2)
    to[bend left=10] cycle;
  \node at (q+-a) [xshift=-2mm,yshift=-4mm] {$x_2$};

  %%%%%%%%%%%%%%% transition %%%%%%%%%%%%%%%%%%%
  \draw[trans] (q+-) to 
    node[pos=0.9,right,xshift=0mm,yshift=2mm] {$\{2\}$}
    (q+-b)
    to[bend left]
    coordinate[pos=0.3] (q1b1)
    (r0);

  \draw[translow] (q+-b)
    to[bend left]
    coordinate[pos=0.6] (q1b2)
    (r+-);

  \filldraw[blueark] (q+-b)
    to[bend left=10] (q1b1)
    to[bend right=30] (q1b2)
    to[bend right=15] cycle;
  \node at (q+-b) [xshift=2mm,yshift=-4mm] {$x_2$};

  %%%%%%%%%%%%%%% transition %%%%%%%%%%%%%%%%%%%
  \draw[trans] (root) to (p);
  \draw[stick] (r+) to node[right,xshift= 1mm] {$\{1,2\}$} (r+a);
  \draw[stick] (r0) to node[left, xshift=-1mm,yshift=1mm] {$\{1,2\}$} (r0a);
  \draw[stick] (r+-) to node[left,xshift=-1mm] {$\{1\}$} (r+-a);
  \draw[stick] (r+-) to node[right,xshift=1mm] {$\{2\}$} (r+-b);
  
\end{tikzpicture}
 }

\caption{LSTA}
\label{fig:bellstates}
\end{subfigure}
% \hfill
\begin{subfigure}{0.34\textwidth}

\begin{mdframed}\centering
\begin{Verbatim}[fontsize=\scriptsize]
Constants
c+ := 1 / sqrt2
c0 := 0
c- := -1 / sqrt2
Root States p
Transitions
[1,{1}](q+, q+-) -> p
[2,{1}](r+, r0) -> q+
[2,{2}](r0, r+) -> q+
[2,{1}](r0, r+-) -> q+-
[2,{2}](r+-, r0) -> q+-
[c+,{1,2}] -> r+
[c0,{1,2}] -> r0
[c+,{1}] -> r+-
[c-,{2}] -> r+-
\end{Verbatim}

\end{mdframed}\vspace{-0.8\baselineskip}

\caption{The \texttt{.lsta} file.}
\label{lsta:fig4}
\end{subfigure}\vspace{-0.5\baselineskip}
\caption{The \lsta for Bell states and its textual description}
\end{figure}

The previous example does not fully demonstrate why incorporating a set of choices (the numbers in the curly brackets) into the design of \lstas is beneficial. Let us consider another well-known example: the set of Bell states $\{\ket{00}\pm\ket{11}, \ket{01}\pm\ket{10}\}$, generated by the \lsta in~\cref{fig:bellstates}. Without the restriction that all transitions at the same level must share a common choice, this \lsta would generate eight different trees (since $q_+$, $q_\pm$, and $r_\pm$ each have two outgoing transitions), which correspond to the quantum states $\{\ket{00}\pm\ket{11}, \ket{01}\pm\ket{10}, \ket{00}\pm\ket{10}, \ket{01}\pm\ket{11}\}$. However, only four of these trees correspond to the Bell states, meaning the others are unintended. The \lsta uses the labeled choices to rule out the unintended trees.
More specifically, at level~2, the two transitions labeled $\{1\}$ can be used
simultaneously, as they share the common choice~$1$.
Similarly, the two transitions labeled $\{2\}$ can be used together due to
their common choice~$2$.
In contrast, a~transition labeled $\{1\}$ cannot be used alongside one labeled
${2}$, as their choice sets are disjoint.
At level~3, the transitions from $r_\pm$ can be used freely with those
from~$r_+$ and~$r_0$, since $\emptyset \subseteq \{1\}$ and $\{2\} \subseteq
\{1,2\}$.
There are two valid combinations of transitions at levels~2 and~3, and this
\lsta generates exactly the four Bell states using the nine transitions shown
in the figure.
The corresponding \texttt{.lsta} file, which illustrates the input format for
\autoqq, is shown in~\cref{lsta:fig4}.
In \texttt{.lsta} files, transitions are labeled with pairs 
$[a,b]$, where $a$ indicates the symbol $x_a$ and $b$ is the set of
choices.

% \fvset{commandchars=\\\{\}}

% \hide{
% \begin{algorithm*}
%   \caption{Apply $-X$ gate on the second qubit, if statement version}
%   \label{alg:if}
%   $\{\specpair{10}{a_0}+ \specpair{11}{a_1}\}$\;
% %  $q := \ket{1}\bigotimes \ket{\phi}$\;
%   $H_1;CX^1_2$\;
%   \lIf{$M_1=0$}{$X_1$}
%   $\{(\specpair{10}{a_0}+ \specpair{11}{a_1}+ \specdef{0})\vee (\specpair{10}{\mathrm{-}a_1}+ \specpair{11}{\mathrm{-}a_0}+ \specdef{0})\}$
% \end{algorithm*}}

% \hide{
% \begin{algorithm*}
%   \caption{Apply $-X$ gate on the second qubit, loop statement version}
%   \label{alg:loop}
%   $\{\specpair{10}{a_0}+ \specpair{11}{a_1}\}$\;
% %  $q := \ket{1}\bigotimes \ket{\phi}$\;
%   $H_1;CX^1_2$\;
% Inv: $\{(\specpair{00}{\frac{a_0}{\sqrt{2}}}+ \specpair{01}{\frac{a_1}{\sqrt{2}}}+ \specpair{10}{\frac{-a_1}{\sqrt{2}}}+ \specpair{11}{\frac{-a_0}{\sqrt{2}}})\}$\;
%   \lWhile{$M_1=0$}{$X_1;H_1;CX^1_2$}
% $\{(\specpair{10}{\mathrm{-}a_1}+ \specpair{11}{\mathrm{-}a_0}+ \specdef{0})\}$
% \end{algorithm*}}

\newcommand{\ifalg}{%
\begin{algorithm*}[H]
  {\color{gray} Pre: $\{\specpair{10}{a_0}+ \specpair{11}{a_1}\}$}\;
  $H_1;CX^1_2$\;
  \lIf{$M_1=0$}{$X_1$}
  {\color{gray} Post: $\{\specpair{10}{a_0}+ \specpair{11}{a_1}\vee (\specpair{10}{\mathrm{-}a_1}+ \specpair{11}{\mathrm{-}a_0}+ \specdef{0})\}$}
\end{algorithm*}}

\newcommand{\loopalg}{%
\begin{algorithm*}[H]

{\color{gray} Pre: $\{\specpair{10}{a_0}+ \specpair{11}{a_1}\}$}\;
%  $q := \ket{1}\bigotimes \ket{\phi}$\;
  $H_1;CX^1_2$\;
{\color{gray} Inv: $\{(\specpair{00}{\frac{a_0}{\sqrt{2}}}+ \specpair{01}{\frac{a_1}{\sqrt{2}}}+ \specpair{10}{\frac{-a_1}{\sqrt{2}}}+ \specpair{11}{\frac{-a_0}{\sqrt{2}}})\}$}\;
  \lWhile{$M_1=0$}{$X_1;H_1;CX^1_2$}
{\color{gray} Post: $\{(\specpair{10}{\mathrm{-}a_1}+ \specpair{11}{\mathrm{-}a_0}+ \specdef{0})\}$}
\end{algorithm*}}

% \begin{figure*}
% \vspace{-1cm}
%     \begin{subfigure}{.5\textwidth}
%         \caption{If version}
%         \ifalg
%         \label{alg:if}
%     \end{subfigure}
%     \begin{subfigure}{.5\textwidth}
%         \caption{Loop version}
%         \loopalg
%         \label{alg:loop}
%     \end{subfigure}
% \caption{Apply $-X$ gate on the second qubit}
% \vspace{-0.8cm}
%
% \end{figure*}

%%%%%%%%%%%%%%%%%%%%%%%%%%%%%%%%%%%%%%%%%%%%%%%%%%%%%%%%%%%%%%%%%%%%%%%%%%%%%%%%
\vspace{-3.0mm}
\section{System Architecture}\label{sec:arch}
\vspace{-2.0mm}
%%%%%%%%%%%%%%%%%%%%%%%%%%%%%%%%%%%%%%%%%%%%%%%%%%%%%%%%%%%%%%%%%%%%%%%%%%%%%%%%
\newcommand{\figArch}[0]{
\begin{figure}[t]%[16]{r}{85mm} %0.63\textwidth
%\vspace{-7mm}
\centering
%\begin{minipage}{85mm}
\scalebox{1}{\begin{tikzpicture}[->,>=stealth,scale=1,transform shape,semithick]
\tikzstyle{every state}=[inner sep=3pt,minimum size=5pt]

  \tikzstyle{component}=[draw,rounded corners=2mm, inner sep=8pt,text width=5.5cm,align=center, node distance=1.5cm,fill=white]

  \node[component] (s2ta) {\textsf{Preprocessor}};
  \node[above of=s2ta] (title) {\large\autoqq};
  \node[component, below = 0.6 cm of s2ta,text width=3.6cm] (exec) {
\textsf{Program Executor~\cite{chen2023autoq}\\(\cref{sec:branch})}\\[2mm]
  \textsf{Invariant Checker\\(\cref{sec:loop})}}; %~\cite{tech}
  \node[component, below = 0.6 cm of exec] (incl) {\textsf{Entailment
  Checker~(\cref{sec:inclusion})}};

  \node[left of=exec,text width=3cm, align=left, node distance=4.6cm, yshift=10mm] (input) {\vbox{
    \textbf{Precondition:}\\
    $\autp$\texttt{.lsta}\\% or $\autp$\texttt{.hsl}\\
    %$\varphi_\autp$\texttt{.smt}\\
    \vspace{0.2cm}%\ \\
    \textbf{Quantum \\Program:}\\
    $\mathit{Prog}$\texttt{.qasm}\\\vspace{0.2cm}%\ \\[2mm]
    \textbf{Postcondition:}\\
    $\autq$\texttt{.lsta}\\\vspace{0.2cm}% or $\autq$\texttt{.hsl}\\
    %$\varphi_\autq$\texttt{.smt}
    \textbf{Loop Invariants:}\\
    $\auti_L$\texttt{.lsta}\\
    %$\varphi_L$\texttt{.smt} 
    for each loop $L$
  }};
  % \node[component, below of=incl,text width=1.5cm,
  % xshift=-1.8cm] (ext) {\textsf{Z3}~\cite{de2008z3}};
  \node[component, left of=incl,text width=1.5cm,
  xshift=-3.3cm] (ext) {\textsf{Z3}~\cite{de2008z3}};

  % \node[right of=s2ta, text width=3cm, align=left, node distance=5.5cm] (s2ta-output) {\vbox{
  %   \textbf{Loop Invariants:}\\
  %   $\auti_L$\texttt{.\{lsta|hsl\}}\\
  %   $\varphi_L$\texttt{.smt} for each while loop $L$\\
  % }};

  %\node[right of=exec,text width=1.5cm, align=left, node distance=4.3cm] (exec-output) {$\autr$\texttt{.lsta}};
  % \node[below of=incl,text width=1.6cm, align=left, xshift=1.57cm, yshift=-5mm] (incl-output) {Verified/Failed};
  \node[below of=incl, yshift=-1mm] (incl-output) {Verified/Failed};

  % \draw[->] ($(s2ta.south)-(2.25cm,0)$)
  % --($(incl.north)-(2.25cm,0)$)node[pos=0.85,left]{$\autq,\varphi_\autq$};
  \draw[->] ($(s2ta.south)-(2.25cm,0)$)
  --($(incl.north)-(2.25cm,0)$)node[midway,left]{$\autq$};
  % \draw[<-] (s2ta.east) -- ++(+10mm,0);
  \draw[<-] (s2ta.west) -- ++(-10mm,0);
  \draw[->] (s2ta.south) -- (exec)node[midway,right]{$\autp, \mathit{Prog}$};
  \draw[->] (exec.south) -- (incl)node[midway,right]{$\autr$};
  % \draw[<->] (incl.194) -- (ext);
  \draw[<->] (incl) -- (ext);
  \draw[->] (incl) -- (incl-output);
  % \draw[->] (incl.344) -- (incl-output);
  %\draw[->] (exec) -- (exec-output);

  \begin{pgfonlayer}{background}
    \draw[-,dashed,rectangle,fill=gray!20,draw=black!70,rounded corners=5pt,inner sep=2pt]
      ([xshift=-4mm]s2ta.west) |- ([yshift=10mm]s2ta.north) -| ([xshift=4mm,yshift=0mm]incl.east) |- ([yshift=-2mm]incl.south west) -| ([xshift=-4mm]s2ta.west);
  \end{pgfonlayer}

\end{tikzpicture}}
%\end{minipage}
%\vspace{-4mm}
\caption{The architecture of \autoqq}
\label{fig:arch}
\end{figure}
}
 %%%%%%%%%%%%%%%%%%%%%%%%%%%%%%%%
We present the architecture of \autoqq in~\cref{fig:arch}. The tool is written
in C++ and uses the SMT solver~\ziii~\cite{de2008z3} for satisfiability and
entailment checking of constraints. We
allow the use of any theory that is supported by~\ziii. In our experiments, we
used \texttt{NIRA} (non-linear integer and real arithmetic).
While this logic is generally undecidable, \ziii~always quickly solved the
formulae we presented to it in our experiments.

\begin{figure}[t]%[16]{r}{85mm} %0.63\textwidth
%\vspace{-7mm}
\centering
%\begin{minipage}{85mm}
\scalebox{1}{\begin{tikzpicture}[->,>=stealth,scale=1,transform shape,semithick]
\tikzstyle{every state}=[inner sep=3pt,minimum size=5pt]

  \tikzstyle{component}=[draw,rounded corners=2mm, inner sep=8pt,text width=5.5cm,align=center, node distance=1.5cm,fill=white]

  \node[component] (s2ta) {\textsf{Preprocessor}};
  \node[above of=s2ta] (title) {\large\autoqq};
  \node[component, below = 0.6 cm of s2ta,text width=3.6cm] (exec) {
\textsf{Program Executor~\cite{chen2023autoq}\\(\cref{sec:branch})}\\[2mm]
  \textsf{Invariant Checker\\(\cref{sec:loop})}}; %~\cite{tech}
  \node[component, below = 0.6 cm of exec] (incl) {\textsf{Entailment
  Checker~(\cref{sec:inclusion})}};

  \node[left of=exec,text width=3cm, align=left, node distance=4.6cm, yshift=10mm] (input) {\vbox{
    \textbf{Precondition:}\\
    $\autp$\texttt{.lsta}\\% or $\autp$\texttt{.hsl}\\
    %$\varphi_\autp$\texttt{.smt}\\
    \vspace{0.2cm}%\ \\
    \textbf{Quantum \\Program:}\\
    $\mathit{Prog}$\texttt{.qasm}\\\vspace{0.2cm}%\ \\[2mm]
    \textbf{Postcondition:}\\
    $\autq$\texttt{.lsta}\\\vspace{0.2cm}% or $\autq$\texttt{.hsl}\\
    %$\varphi_\autq$\texttt{.smt}
    \textbf{Loop Invariants:}\\
    $\auti_L$\texttt{.lsta}\\
    %$\varphi_L$\texttt{.smt} 
    for each loop $L$
  }};
  % \node[component, below of=incl,text width=1.5cm,
  % xshift=-1.8cm] (ext) {\textsf{Z3}~\cite{de2008z3}};
  \node[component, left of=incl,text width=1.5cm,
  xshift=-3.3cm] (ext) {\textsf{Z3}~\cite{de2008z3}};

  % \node[right of=s2ta, text width=3cm, align=left, node distance=5.5cm] (s2ta-output) {\vbox{
  %   \textbf{Loop Invariants:}\\
  %   $\auti_L$\texttt{.\{lsta|hsl\}}\\
  %   $\varphi_L$\texttt{.smt} for each while loop $L$\\
  % }};

  %\node[right of=exec,text width=1.5cm, align=left, node distance=4.3cm] (exec-output) {$\autr$\texttt{.lsta}};
  % \node[below of=incl,text width=1.6cm, align=left, xshift=1.57cm, yshift=-5mm] (incl-output) {Verified/Failed};
  \node[below of=incl, yshift=-1mm] (incl-output) {Verified/Failed};

  % \draw[->] ($(s2ta.south)-(2.25cm,0)$)
  % --($(incl.north)-(2.25cm,0)$)node[pos=0.85,left]{$\autq,\varphi_\autq$};
  \draw[->] ($(s2ta.south)-(2.25cm,0)$)
  --($(incl.north)-(2.25cm,0)$)node[midway,left]{$\autq$};
  % \draw[<-] (s2ta.east) -- ++(+10mm,0);
  \draw[<-] (s2ta.west) -- ++(-10mm,0);
  \draw[->] (s2ta.south) -- (exec)node[midway,right]{$\autp, \mathit{Prog}$};
  \draw[->] (exec.south) -- (incl)node[midway,right]{$\autr$};
  % \draw[<->] (incl.194) -- (ext);
  \draw[<->] (incl) -- (ext);
  \draw[->] (incl) -- (incl-output);
  % \draw[->] (incl.344) -- (incl-output);
  %\draw[->] (exec) -- (exec-output);

  \begin{pgfonlayer}{background}
    \draw[-,dashed,rectangle,fill=gray!20,draw=black!70,rounded corners=5pt,inner sep=2pt]
      ([xshift=-4mm]s2ta.west) |- ([yshift=10mm]s2ta.north) -| ([xshift=4mm,yshift=0mm]incl.east) |- ([yshift=-2mm]incl.south west) -| ([xshift=-4mm]s2ta.west);
  \end{pgfonlayer}

\end{tikzpicture}}
%\end{minipage}
%\vspace{-4mm}
\caption{The architecture of \autoqq}
\label{fig:arch}
\end{figure}

Similar to verifiers for classical programs, in order to use \autoqq, the user
needs to provide the following:
%
% \begin{itemize}
\begin{inparaenum}[(i)]
    \item a~\emph{quantum program} in the \openqasm 3.0 format,
    \item \emph{pre- and post-conditions} in the \texttt{.lsta} format along with SMT formulae in the \texttt{.smt} format to constrain the terms, and
    \item \emph{invariant} for each loop in the \texttt{.lsta} format together with an SMT formula.
\end{inparaenum}
% \end{itemize}
%
The specification and invariant for each loop can be written as an \lsta
(an~\texttt{.lsta} file).
Once these files are provided, \autoqq will process them and report either
``Verified'' or ``Failed''.

Compared to \autoq, there are several major changes in \autoqq.
Firstly, it features a new input interface to facilitate the use of quantum
programs (instead of only circuits) and uses \lsta as the back-end model (instead of standard tree automata).
Additionally, \textsf{Program Executor} now supports measurement and
branch statements.
Another significant addition is the new \textsf{Invariant
Checker} component, which handles loop invariants.
The \textsf{Invariant Checker}
also uses the \textsf{Entailment Checker} to verify the \emph{inductiveness} of
% the invariant, which we do not explicitly show in the figure.
the invariant, which we do not explicitly show in the figure.

%%%%%%%%%%%%%%%%%%%%%%%%%%%%%%%%%%%%%%%%%%%%%%%%%%%%%%%%%%%%%%%%%%
\newcommand{\algMeasurement}[0]{
\begin{algorithm}[t]
\SetNoFillComment
\KwIn{\lsta $\aut=\tuple{Q, \Sigma, \Delta, \rootstates, \globconstr}$, target qubit $x_i$, measurement outcome $b$}
\KwOut{The \lsta $M_i^{=b}(\aut)$}
$Q':=\{q'\mid q\in Q\};\ \rootstates':=\{q'\mid q\in \rootstates\}$\;
$\ctr{\Delta'}{\ell'}:= \{\ctranstreenoset {q'} f {q_1',q_2'}{C} \mid \ctranstreenoset {q}{f}{q_1, q_2}{C}\in \ctr{\Delta}{\ell}\} \cup \{\ctranstreenoset{q'}{0}{}{C} \mid \ctranstreenoset{q}{f}{}{C}\in \ctr{\Delta}{\ell}\}$\;
  
% \DontPrintSemicolon
% \tcc*[l]{making one primed copy of~$\aut$ whose leaf symbols are $0$}
% \PrintSemicolon
$\Delta^{\neq{x_i}} := \{\ctranstreenoset{q}{f}{q_0,q_1}{C} \mid \ctranstreenoset{q}{f}{q_0,q_1}{C}\in \Delta \wedge f\neq x_i\}$\; 
\lIf{$b=0$}{
$\Delta^{={x_i}} := \{\ctranstreenoset{q}{x_i}{q_0,q'_1}{C} \mid \ctranstreenoset{q}{x_i}{q_0,q_1}{C}\in \Delta\}$
}
\hspace*{12mm}\lElse{
\hspace*{1mm}$\Delta^{={x_i}} := \{\ctranstreenoset{q}{x_i}{q'_0,q_1}{C} \mid \ctranstreenoset{q}{x_i}{q_0,q_1}{C}\in \Delta\}$
}
%\DontPrintSemicolon
%\tcc*[l]{updating all $x_i$-labeled transitions $\transtree {q} {x_i} {q_0, q_1}$ to jump to the primed version depending on the value of $b$}     
%\PrintSemicolon
\Return {$M_i^{=b}(\aut)=\tuple{Q\cup Q', \Sigma,
	\Delta^{={x_i}}\cup \Delta^{\neq{x_i}}\cup\Delta', \rootstates, \globconstr}$}\;
\caption{Algorithm for measurement}
\label{alg:measure}
\end{algorithm}
}

%%%%%%%%%%%%%%%%%%%%%%%%%%%%%%%%%%%%%%%%%%%%%%%%%%%%%%%%%%%%%%%%%%%%%%%%%%%%%%%%
 \vspace{-3.0mm}
\section{Handling Branch, Measurement, and Loop}\label{sec:branch}
 \vspace{-2.0mm}
%%%%%%%%%%%%%%%%%%%%%%%%%%%%%%%%%%%%%%%%%%%%%%%%%%%%%%%%%%%%%%%%%%%%%%%%%%%%%%%%
We will start by presenting the syntax of quantum programs that \autoqq can
handle and then informally describe their semantics.
We use a flavor of quantum programs that is similar to the one in
\cite{ying2012floyd}, which is captured by the following grammar:
$$
P~~ ::={}  ~~U~~\bigm| ~~P;P~~ \bigm| ~~\textbf{while\ } (M_i = b)~\textbf{do}~\{P\}
  ~~\bigm| ~~\textbf{if\ } (M_i = b)~\textbf{then}~\{P\}~\textbf{else}~\{P\}
$$
where $P$ is a quantum program, $U$ is a quantum gate annotated with its
control and target qubits (e.g., $CX_1^2$), $b\in\{0,1\}$, and $M_i$ is the measured value of
the $i$-th qubit.
%, and each $P$ on the right-hand side is a new instance of quantum programs, defined recursively.
\autoqq supports
standard non-parameterized quantum gates that allow (approximate) universal
computation~\cite{BoykinMPRV00,Aharonov03}, including Clifford gates ($H$, $S$,
and $\textit{CX}$), $T$, $Z$, $\mathit{SWAP}$, Toffoli, etc.
% several sets of (approximate) universal quantum gates, including
%
% \begin{inparaenum}[(i)]
%   \item  Clifford gates ($H$, $S$, and $\textit{CX}$) and~$T$
%     (cf.~\cite{BoykinMPRV00}) or
%   \item  Toffoli and~$H$ (cf.~\cite{Aharonov03}).
% \end{inparaenum}

The execution of a quantum gate $U$ updates a quantum state (tree) in the
standard way~\cite{ChenCLLTY23}.
The language allows sequential composition ($P; P$) of gate operations,
branches (\textbf{if}\ \ldots \textbf{else}\ \ldots), and loops (\textbf{while}\ \ldots).
When using \textbf{if} and \textbf{while} statements, the condition $M_i = b$
(denoting that the value obtained from measuring~$x_i$ was~$b$) is used
to determine in which path to continue.
% The condition denotes that the value obtained from measuring~$x_i$ is~$b$.
% We go to the positive branch when the measured value is $b$, and to the negative
% branch otherwise. 

%\ol{we somehow jump from talking about the semantics as an execution of one quantum state to TAs} \yfc{addressed, please check}
%*******************************************************************************
\vspace{-0.0mm}
\subsection{Handling Measurement}
\vspace{-0.0mm}
%*******************************************************************************

\algMeasurement  %%%%%%%%%%%%%%

The key part of handling branch statements in \autoqq is how
measurement changes the quantum states and how we should update the \lsta
capturing the set of reachable states.
As mentioned in~\cref{sec:overview}, if the
measured value of $x_i$ is $1$, then we should update the probability of the
$0$-subtrees below all $x_i$ nodes to $0$.
Examples can be found
in~\cref{fig:H1CX12state} (before measuring~$x_1$) and~\cref{fig:M0state,fig:M1state} (after
measuring~$x_1$ as~$0$ and~$1$ respectively).
An important design choice was that we do not normalize the
probability amplitudes and simply just make all leaves of one of the subtrees
zero in this step, leaving the task of \emph{matching non-normalized states} to the
 entailment test (cf.\ \cref{sec:inclusion}). 

In some cases, the measurement can generate an \lsta whose language contains
a~tree where all leaves are~$0$.
This can happen, e.g., when we compute the tree representing the quantum state
obtained from the state in \cref{fig:initstate} by measuring $x_1 = 0$.
We do not consider such a tree to represent a~quantum state.
To handle such cases, our entailment test $\autr \entailedbyuptosc \autq$ (formally defined in
\cref{sec:inclusion}) adds a \emph{$0$-labeled tree} to the language
of $\autq$ before the test.
We use $M_i^{=b}(\aut)$ to represent the \lsta obtained from $\aut$ after
measuring~$x_i$ for the outcome~$b$. 
%We give technical details of the construction of $M_i^{=b}(\aut)$ in \cref{sec:measurement}.
The procedure for computing $M_i^{=b}(\aut)$ is given in~\cref{alg:measure}.
% The formal definition of $\autr \entailedbyuptosc \autq$ will be given in the next section.

The goal of the algorithm is to update all leaf values of $\bar{b}$-subtrees
under $x_i$ to $0$, where $\bar{b} = 1 - b$.
%converts $0$ to $1$ and $1$ to $0$.
Lines~1 and~2 of~\cref{alg:measure} create a primed copy of the input LSTA and update all leaf values to $0$ (Line~2). Lines~3 to~5 construct the new transition system: only transitions labeled with $x_i$ are modified (Lines~4 and~5), while others remain unchanged (Line 3).
The key steps are in Lines~4 and~5, which control
the subtrees of the measured qubit.
% how subtrees transition to their primed versions.
When $b = 0$ (Line~4), all leaves of the $1$-subtree are modified to $0$, and thus, we update $q_1$ in the original transition to $q_1'$ (symmetrically for $b=1$ in Line~5).
% The case for $b = 1$ (Line~5) is symmetric.

%*******************************************************************************
\vspace{-0.0mm}
\subsection{Handling Branch Statements}
\vspace{-0.0mm}
%*******************************************************************************

Given an \lsta $\aut$ that recognizes a set of quantum states, we can
precisely compute the set of states that result from executing a branch
statement \textbf{if\ } ($M_i = b$) \textbf{then} $\{P_1\}$ \textbf{else}
$\{P_0\}$ as follows (assuming that $P_0$ and $P_1$ do not involve loops):
\begin{inparaenum}[(i)]
    \item Create two \lstas $M_i^{=1}(\aut)$ and $M_i^{=0}(\aut)$.
    \item Compute the result after executing $P_b$ from $M_i^{=b}(\aut)$ for
      $b\in \{0,1\}$, following the gates' semantics and recursively trigger the
      procedure for branches.
      We use $\aut^0$ and $\aut^1$ to denote the \lstas after executing $P_0$ and
      $P_1$, respectively.
    \item Construct an \lsta recognizing the union of $\aut^0$ and $\aut^1$ and return it as the final result of this procedure. 
\end{inparaenum}
In principle, our approach can handle nested control flow.
We are, however, not aware of any real-world quantum program that uses a nested
control structure, and, therefore, for simplicity, \autoqq now only supports
programs with non-nested control flow.

%*******************************************************************************
\vspace{-0.0mm}
\subsection{Handling Loop Statements}\label{sec:loop}
\vspace{-0.0mm}
%*******************************************************************************

If we come across a loop statement \textbf{while} ($M_i = b$) \textbf{do} $\{B\}$ with $B$
being the loop body, we require the user to provide a loop invariant in the form
of an \lsta. We refer to the invariant as $\auti$; it needs to satisfy the following properties:
\begin{inparaenum}[(i)]
  \item  It contains all reachable states, captured by an \lsta $\autr$, before
    entering the loop. That is, $\autr\entailedbyuptosc \auti$.
  \item It is \emph{inductive}, i.e., $B(M_i^{=b}(\auti))\entailedbyuptosc
    \auti'$, where $B(\aut)$ denotes an \lsta recognizing the set of quantum states after executing $B$ from the quantum states in $\aut$ and $\auti'$ is an \lsta obtained from $\auti$ whose variables and constraints are updated to a~primed version.
\end{inparaenum}
The inductiveness guarantees that if we take any state accepted by $\auti$
and perform~$B$ on the state, the result will also be accepted by $\auti$.
Together with the condition that~$\auti$ covers all reachable states before
entering the loop, $\auti$~\emph{over-approximates} all reachable states at the
loop entrance.
We can, therefore, use $M_i^{\neq b}(\auti)$ to over-approximate all reachable states at
the loop exit. 

%%%%%%%%%%%%%%%%%%%%%%%%%%%%%%%%%%%%%%%%%%%%%%%%%%%%%%%%%%%%%%%%%%%%%%%%%%
\newcommand{\algEntailment}[0]{
\begin{algorithm}[t]
  \SetKw{Continue}{continue}
  \caption{Checking $\autr \entailedbyuptosc \autq$ \wl{concrete version}}
  \label{alg:InclCheckConcrete}
  \KwIn{A TA $\autr = (S_\R, \Delta_\R, F_\R)$, a~TA
    $\autq = (S_\Q, \Delta_\Q, F_\Q)$}
  \KwOut{$\mathit{true}$ if $\autr \entailedbyuptosc \autq$, $\mathit{false}$ otherwise}

  $\processed \gets \emptyset$\;
  $\worklist \gets \{(s_\R, U_\Q) \mid \transleaf{s_\R}{v_\R} \in \Delta_\R,$\\
  \hspace*{8.7mm} $U_\Q = \{(u_\Q, c) \in S_\Q \times \realsext \mid
  \transleaf{u_\Q}{v_\Q} \in \Delta_\Q \land (v_\R \neq 0 \Rightarrow (v_\Q = c
  \cdot v_\R \land c \neq 0)) \land {}$\\[1mm]
  \hspace*{66.5mm} $(v_\R = 0 \Rightarrow (v_\Q = 0 \land c = ?))\}\}$\label{ln:workset}\;
  \While{$\worklist \neq \emptyset$}{
    $(s_\R, U_\Q) \gets \mathit{Worklist}.\mathit{pop}()$\label{ln:pop}\;
    \lIf{$s_\R \in F_\R \land \forall (u_q, c) \in U_\Q\colon u_q \notin F_\Q$}{\label{ln:roottest}
      \Return{$\mathit{false}$}
    }
    $\processed \gets \processed \cup \{(s_\R, U_\Q)\}$\;\label{ln:processed}

    $\mathit{tmp} \gets (\{(s_\R, U_\Q)\} \times \processed) \cup (\processed
      \times \{(s_\R, U_\Q)\})$\label{ln:tmp}\;
    \ForEach{$((s_\R^1, U_\Q^1), (s_\R^2, U_\Q^2)) \in \mathit{tmp}, \alpha \in
    \vars$}{
      $H_\R \gets \{s_\R' \in S_\R \mid \transtree{s_\R'}{\alpha}{s_\R^1, s_\R^2} \in \Delta_\R \}$\;
      $U_\Q' \gets \emptyset$\;
      \ForEach{$(s_\Q^1, c^1) \in U_\Q^1, (s_\Q^2, c^2) \in U_\Q^2$}{
        \lIf{$c^1 \neq c^2 \land c^1 \neq ? \land c^2 \neq
        ?$}{\label{ln:connecttrees}
          \Continue
        }
        \ForEach{$\transtree{s_\Q}{\alpha}{s_\Q^1, s_\Q^2} \in \Delta_\Q$}{
          $U_\Q' \gets U_\Q' \cup \{(s_q, c) \mid c = c^2 \textbf{ if } c^1 = ? \textbf{ otherwise } c^1\}$\;
        }
      }
      % $U_\Q' \gets \{(s_\Q, c) \in Q_\Q \times \realsext \mid \exists (s_\Q^1, c^1) \in U_\Q^1,\ \exists (s_\Q^2, c^2) \in
      %   U_\Q^2\colon$\\
      %   $\transtree{s_\Q}{\alpha}{s_\Q^1, s_\Q^2} \in \Delta_\Q \land (c^1 = ? \Rightarrow c = c^2) \land (c^2 = ? \Rightarrow c = c^1) \land (c^1 = c^2 \Rightarrow c = c^1\}$\;
      \ForEach{$s_\R' \in H_\R$ s.t.\ $(s_\R', U_\Q') \notin \processed \cup \worklist$}{
        \label{ln:downcltest}
        $\worklist \gets \worklist \cup \{(s_\R', U_\Q')\}$\label{ln:addnew}\;
      }
    }
  }
  \Return{$\mathit{true}$}\;
\end{algorithm}
}

\newcommand{\algEntailmentsym}[0]{
\begin{figure}[t]
\begin{algorithm}[H]
    \SetKw{Continue}{continue}
    \caption{Checking $\autr \entailedbyuptosc \autq$}
    \label{alg:InclCheck}
    \KwIn{TAs $\autr = (S_\R, \Sigma\cup\vars, \Delta_\R, F_\R, \varphi_\R)$ and
        $\autq = (S_\Q, \Sigma\cup\vars, \Delta_\Q, F_\Q, \varphi_\Q)$}
    % \KwIn{A TA $\autr = (S_\R, \Delta_\R, F_\R)$, a~TA
    %     $\autq = (S_\Q, \Delta_\Q, F_\Q)$, global constraint~$\globconstr$
    %     \ol{old}}
    \KwOut{$\mathit{true}$ if $\autr \entailedbyuptosc \autq$, $\mathit{false}$ otherwise}

    $\processed \gets \emptyset$\;
    $\worklist \gets \big\{(s_\R, Y, U_\Q) \mid \transleaf{s_\R}{t_\R} \in
    \Delta_\R, Y = \varsof{t_\R}, $\label{ln:workset}\\[-1mm]
    \hspace*{34mm}$U_\Q = \{(u_\Q, \formulaof{t_r \cdot \scale = t_q}) \in S_\Q \times \setform \mid \transleaf{u_\Q}{t_\Q} \in \Delta_\Q\}\big\}$\;
    % \land \issatof{\formulaof{t_\R \cdot \scale = t_\Q \land \globconstr}}
    \While{$\worklist \neq \emptyset$}{
      $(s_\R, Y, U_\Q) \gets \mathit{Worklist}.\mathit{pop}()$\;
        % \lIf{$s_\R \in F_\R \land \ol{fix}\forall (u_q, \psi) \in U_\Q\colon u_q \notin F_\Q$}{\label{ln:roottest}
        \lIf{$s_\R \in F_\R \land 
        % \issatof{\formulaof{\neg\bigvee\limits_{\begin{array}{c}
        %   \scriptstyle (u_q, \psi) \in U_q\\
        %   \scriptstyle u_q \in F_\Q
        % \end{array}}
        % \exists x_1, \ldots, x_n \in \complex\, \exists \scale \in \reals \colon
        % \varphi_\R \land \neg (\varphi_\Q \land \psi)
        % \issatof{\formulaof{\neg\bigvee\big\{
        % \exists x_1, \ldots, x_n \in \complex\, \exists \scale \in \reals \colon
        % \varphi_\R \land \neg (\varphi_\Q \land \psi) \mid 
        %   \scriptstyle (u_q, \psi) \in U_q,
        %   \scriptstyle u_q \in F_\Q\big\}
        \issatof{\formulaof{\neg\bigvee_{\psi \in \Psi}
        \exists x_1, \ldots, x_n \in \complex\, \exists \scale \in \reals \colon
        % \varphi_\R \land \neg (\varphi_\Q \land \psi)
        \varphi_\R \Rightarrow (\varphi_\Q \land \psi)
        }}$ \label{ln:roottest}
        for $\Psi = \{\psi \mid (u_q, \psi) \in U_q, u_q \in F_\Q\}$ and $\{x_1,
        \ldots, x_n\} = \vars \setminus Y$}{
          \Return{$\mathit{false}$}
        }
        $\processed \gets \processed \cup \{(s_\R, Y, U_\Q)\}$\;\label{ln:processed}

        $\mathit{tmp} \gets (\{(s_\R, Y, U_\Q)\} \times \processed) \cup (\processed
        \times \{(s_\R, Y, U_\Q)\})$\label{ln:tmp}\;
        \ForEach{$((s_\R^1, Y^1, U_\Q^1), (s_\R^2, Y^2, U_\Q^2)) \in \mathit{tmp}, \alpha \in
        \Sigma$}{
        $H_\R \gets \{s_\R' \in S_\R \mid \transtree{s_\R'}{\alpha}{s_\R^1,
        s_\R^2} \in \Delta_\R \}$; $Y' = Y^1 \cup Y^2$\label{ln:variables}\;
        $U_\Q' \gets \emptyset$\;
        \ForEach{$(s_\Q^1, \psi^1) \in U_\Q^1, (s_\Q^2, \psi^2) \in U_\Q^2$}{
          % \lIf{$\neg\issatof{\ol{fix}\formulaof{\psi^1 \land \psi^2 \land \globconstr}}$}{\Continue}
          \lForEach{$\transtree{s_\Q}{\alpha}{s_\Q^1, s_\Q^2} \in \Delta_\Q$}{
            $U_\Q' \gets U_\Q' \cup \{(s_q, \formulaof{\psi^1 \land \psi^2})\}$
            \label{ln:conjoin}
          }
        }
        % $U_\Q' \gets \{(s_\Q, c) \in Q_\Q \times \realsext \mid \exists (s_\Q^1, c^1) \in U_\Q^1,\ \exists (s_\Q^2, c^2) \in
        %   U_\Q^2\colon$\\
        %   $\transtree{s_\Q}{\alpha}{s_\Q^1, s_\Q^2} \in \Delta_\Q \land (c^1 = ? \Rightarrow c = c^2) \land (c^2 = ? \Rightarrow c = c^1) \land (c^1 = c^2 \Rightarrow c = c^1\}$\;
        \ForEach{$s_\R' \in H_\R$}{
          \label{ln:downcltest}
          \lIf{$\exists(s_\R', Y', V) \in \processed \cup \worklist\colon$\\
          \hspace*{2.3mm} $\forall (u_q, \psi) \in U_\Q'\colon \exists
          (u_q, \chi) \in V\colon \psi \Rightarrow \chi$}{
            \Continue
          }
          $\worklist \gets \worklist \cup \{(s_\R', Y', U_\Q')\}$\label{ln:addnew}\;
        }
        % \ForEach{$s_\R' \in H_\R$}{
        %   \ForEach{$(s_\R', V) \in \processed \cup \worklist$}{
        %     \label{ln:downcltest}
        %     \lIf{$\forall (u_q^i, \psi^i) \in U_\Q' ~ \exists (v^j, \chi^j) \in V\colon \psi^i \Rightarrow \chi^j$}{\Continue}
        %     $\worklist \gets \worklist \cup \{(s_\R', U_\Q')\}$\label{ln:addnew}\;
        %   }
        % }
        % \ol{old}
        % \ForEach{$s_\R' \in H_\R$ s.t.\ $(s_\R', U_\Q') \notin \processed \cup \worklist$}{
        %     \label{ln:downcltest}
        %     $\worklist \gets \worklist \cup \{(s_\R', U_\Q')\}$\label{ln:addnew}\;
        % }
      }
    }
    \Return{$\mathit{true}$}\;
\end{algorithm}
\vspace*{-4mm}
\end{figure}
}

\vspace{-3.0mm}
\section{Testing Entailment up to Scaling}\label{sec:inclusion}
\vspace{-2.0mm}
%%%%%%%%%%%%%%%%%%%%%%%%%%%%%%%%%%%%%%%%%%%%%%%%%%%%%%%%%%%%%%%%%%%%%%%%%%%%%%%%

In our approach, we use a~special entailment test at some points, which we
call \emph{entailment up to scaling}, denoted as $\aut \entailedbyuptosc
\autb$.
The reason for a~special entailment relation is that---as mentioned before---at
measurements, we do not perform normalization.
Intuitively, given two \lstas
$\aut=\tuple{Q_\aut,\Sigma,\Delta_\aut,\rootstates_\aut,\globconstr_\aut}$ and
$\autb=\tuple{Q_\autb,\Sigma,\Delta_\autb,\rootstates_\autb,\globconstr_\autb}$,
the relation $\aut \entailedbyuptosc \autb$ holds if and only if for every
tree~$T_\aut$ in the language of~$\aut$ and assignment to the variables occurring in~$T_\aut$, we can find
a~linearly scaled copy of~$T_\aut$ in the semantics of $\autb$ (such that the
values of variables occurring in both~$T_\aut$ and $T_\autb$ match).
% some tree~$T_\autb$ in the language
% of~$\autb$ such that if we multiply all leaves of~$T_\autb$ by some non-zero
% number~$r$, we obtain~$T_\aut$.
Formally,
% $\aut \entailedbyuptosc \autb$ iff
% %
% \begin{equation}
% (\forall T_\aut \in \langof \aut)
% (\forall \sigma_\aut\colon \varsof{T_\aut} \to \complex)
% (\exists T_\autb \in \langof \autb)
% (\exists \sigma_\autb\colon (\varsof{T_\autb}\setminus \varsof{T_\autb}) \to \complex)
% (\exists r \in \reals
% \setminus \{0\})\colon T_\aut =  r \cdot T_\autb,
% \end{equation}
%
\begin{align*}
  \aut \entailedbyuptosc \autb \iff {} &
(\forall T_\aut \in \langof \aut)
(\forall \sigma_\aut\colon \varsof{T_\aut} \to \complex)\colon\\ &
(\exists T_\autb \in \langof \autb)
(\exists \sigma_\autb\colon (\varsof{T_\autb}\setminus \varsof{T_\aut}) \to \complex)\colon\\ &
(\exists r \in \reals \setminus \{0\})\colon
  T_\aut[\sigma_\aut] =  r \cdot T_\autb[\sigma_\aut \cup \sigma_\autb],
\end{align*}
% \begin{equation}
% \aut \entailedbyuptosc \autb \iff
% (\forall T_\aut \in \semof \aut)
% (\exists T_\autb \in \semof \autb)
% (\exists r \in \reals \setminus \{0\})\colon T_\aut =  r \cdot T_\autb,
% \end{equation}
% \begin{equation}
% (\forall \sigma\colon \vars \to \complex)
% (\forall T_\aut \in \langof \aut)
% (\exists T_\autb \in \langof \autb)
% (\exists r \in \reals
% \setminus \{0\})\colon T_\aut =  r \cdot T_\autb,
% \end{equation}
%
where $r \cdot T$
denotes the tree with the same structure as~$T$ with all numbers in
leaves multiplied by~$r$
and $\varsof{\gamma}$ for any mathematical object~$\gamma$ (a term, a tree with
terms in leaves, a~set of terms, etc.) denotes the set of free variables
occurring in the object.
The $\entailedbyuptosc$ relation is central to our approach.

% \algEntailmentsym
% Given two \lstas $\aut=\tuple{Q_\aut,\Sigma,\Delta_\aut,\rootstates_\aut,\globconstr_\aut}$ and $\autb=\tuple{Q_\autb,\Sigma,\Delta_\autb,\rootstates_\autb,\globconstr_\autb}$, we introduce a procedure to check for \emph{entailment up to scaling} $\aut
% \entailedbyuptosc \autb$ in this section. Given two \lstas $\aut$ and $\autb$, we say the entailment up to scaling
% relation $\aut \entailedbyuptosc \autb$ holds if  $\forall \sigma\in \vars \to \complex,\ \forall T_\aut
% \in \langof \aut,\ \exists\ T_\autb \in \langof \autb,\ r \in \reals
% \setminus \{0\}\colon T_\aut =  r \cdot T_\autb$, where $ r \cdot T_\autb$ for
% a tree $T_\autb$ and $r \in \reals$ denotes a tree with the same structure as $T_\autb$ with all numbers in leaves multiplied by $r$.
% Intuitively, the relation says that for all assignments $\sigma$ to $\vars$, $\langof \autb$ contains linearly scaled copies of all trees in $\langof \aut$; the relation is central to our approach. 

Enumerating all trees of~$\aut$ and looking for their scaled copies in~$\autb$
would be too inefficient and even impossible in the case of \lsta{}s with
infinite languages (such as those modelling invariants of parameterized quantum
programs~\cite{popl}).
Therefore, we modified the algorithm for \lsta language inclusion test presented in~\cite{popl}.
We note that language inclusion testing for \lstas is more involved than for
standard TAs (cf.\
\cite{BouajjaniHHTV08,AbdullaCHMV10,HolikLSV11,lengal2012vata}).
In the modification, we allow to relate the leaf values with a linear factor
for scaling (in contrast to only by identity as done in the original
inclusion testing algorithm), so that it tests the entailment $\aut \entailedbyuptosc \autb$.

% Efficiently checking for entailment up to scaling, denoted as $\aut \entailedbyuptosc \autb$, between the two \lstas $\aut$ and $\autb$ requires careful design. We cannot simply enumerate all trees generated by $\aut$ and verify whether each tree can be generated by $\autb$ with a scaling factor, as $\langof{\aut}$ may be an infinite set (for example, when the transition system involves loops).
% To address this, we present a formal algorithm that leverages a key property: all occurrences of a state at the same level in a run generate the same subtree.

The algorithm makes use of the following essential property of trees generated by an \lsta~$\aut$:
if two nodes at the same level of a~tree~$T$ are labelled by the same state in
an accepting run of~$\aut$ on~$T$, then the subtrees rooted in these nodes are
identical (this follows from the semantics of \lsta{}s and the restriction on
transitions, cf.\ \cref{sec:lsta}).

Intuitively, the algorithm works as follows.
It starts in the root states of~$\aut$ and~$\autb$ and performs a~downward
traversal through the \lsta{}s, level by level, remembering, at each level, how
states from~$\aut$ can map to the states in~$\autb$.
Moreover, the algorithm also remembers how the terms in the leaves of the tree
from~$\aut$ map to the terms in the leaves of the tree from~$\autb$.
The downward successors of each level are computed from transitions leaving
states at the level that need to be synchronized on their choice.
The algorithm explores the space of all of the reachable mappings until it
reaches a~point such that the tree from~$\aut$ has all branches terminated.
At this moment, we check that the terms from the leaf transitions of~$\aut$ can
be mapped (up to scaling) to the corresponding terms from the leaf transitions
in~$\autb$, and if not, we can conclude that the entailment does not hold.
% finds a~level at which some leaf transitions of~$\aut$ could not be properly
% mapped (up to scaling) to the leaf transitions in~$\autb$, which signals that the
% entailment does not hold (there is a~tree in the language of~$\aut$ for which
% there is no scaled version in the language of~$\autb$).
% If no such level exists, the entailment up to scaling holds.

\inclusionExample

Formally, the algorithm performs a~search in the directed graph $(V,E)$ (which
is constructed on the fly), where vertices~$V$ are of the form
% $V=(D,\{f_1,\ldots,f_m\})$
$V=(D,\{(f_1, g_1),\ldots,(f_m, g_m)\})$
where $D\subseteq Q_\aut$ is called the
\emph{domain}, each $f_i\colon D \to 2^{Q_\autb}$ is a~map that assigns
every state of~$\aut$ from the domain~$D$ a~set of states of~$\autb$,
% \newtext{and each $g_i\colon \symterm \partialto 2^{\symterm}$ is a~(partial)
and each $g_i\colon \symterm \to 2^{\symterm}$ is a~(partial) mapping
from terms to sets of terms.
% that says how terms occurring in the leaf transitions of~$\aut$ should be
% covered by sets of terms from leaf transitions of~$\autb$.}
Intuitively, $D$~represents the set of states of~$\aut$ at one level of~$\aut$'s
run~$\run$, and every~$f_i$ represents the same level of some possible
run~$\rho_i$ of~$\autb$ on the same tree and the way it can match the
run~$\rho$ of~$\aut$.
For instance, in~\cref{fig:inclusionExample}, the state $q$ of $\aut$
corresponds to the states $r,s$ of $\autb$ because they are used in the same
tree level and the same tree nodes,
so we have $f_i(q) = \{r,s\}$.
Due to the property that all occurrences of a state at the same level in a~run
generate the same subtree mentioned above, we only need to maintain encountered
states and their alignment with each another.
The term mappings~$g_i$ are used to remember how the terms from the
leaves of~$\aut$ are mapped to terms in the leaves of~$\autb$.
For instance, if some term~$t$ from a~leaf of~$\aut$ is mapped by two
different terms~$t_1$ and $t_2$ of~$\autb$, we will later need to check whether
there is a~scaling factor~$r$ such that $t = r\cdot t_1$ and, at the same type,
$t = r \cdot t_2$ (and the global constraints of~$\aut$ and~$\autb$ are
satisfied).
We give a~general algorithm here; for \lstas that accept only perfect trees (as
is the case for the ones encoding quantum states), all branches of the accepted
trees terminate at the same time so there is no need to remember the term
mappings~$g$ across different levels and the scaling compatibility could be
checked only locally.

% The order of their occurrence is irrelevant.

% We can formulate the entailment problem as a directed graph search problem. In this graph $(V,E)$, vertices in $V$ are of the form $v=(D,\{f_1,\ldots,f_m\})$ where $D\subseteq Q_\aut$ is the \emph{domain} and $f_i\colon D \rightarrow 2^{Q_\autb}$ is a map that assigns sets of states of $\autb$ to states of $\aut$ from the domain $D$. 
% Intuitively, $D$~represents the set of states of $\aut$ in a level of an $\aut$'s run $\run$, and every $f_i$ represents the same level of some possible run $\rho_i$ of $\autb$ on the same tree, and how it matches the run $\rho$ of $\aut$. For instance, in~\cref{fig:inclusionExample}, the state $q$ of $\aut$ corresponds to the states $r,s$ of $\autb$ because they are used in the same tree level and the same tree nodes. So we have $f_i(q) = \{r,s\}$. Due to the property that all occurrences of a state at the same level in a run generate the same subtree, we only need to maintain encountered states and their alignment with each another. The order of their occurrence is irrelevant.

The graph search begins from the \emph{source} vertices, one for each state $q
\in \rootstates_\aut$, of the form $(\{q\},\{(\{q\mapsto
\{r_1\}\}, \emptyset),\ldots,(\{q\mapsto \{r_k\}\}, \emptyset)\})$, where $\{r_1,\ldots, r_k \} =
\rootstates_\autb$, corresponding to the root states of both $\aut$ and
$\autb$ (the $\emptyset$'s denote empty term mappings).
If the search finds a \emph{terminal} vertex $(\emptyset, F)$, where $(\emptyset, g)
\notin F$, meaning that an accepting run of $\aut$ has been found, but there
is no corresponding matching run of $\autb$ (for any~$g$), we can conclude that the
entailment test failed (it represents the case when~$\aut$ finished reading all
branches of the tree but~$\autb$ did not).
On the other hand, if there is $(\emptyset, g) \in F$,
we still need to check that the set of terms~$g$ is compatible.
The graph's edges represent generating the next level of runs
for both $\aut$ and $\autb$ and how the respective states align with each
other.
The specific construction of the edges from a vertex $v = (D, \{(f_1, g_1), \ldots,
(f_m,g_m)\})$, where $D \neq \emptyset$, to each lower-level vertex $v' = (D',
\{(f'_1, g'_1), \ldots, (f'_n, g'n)\})$ follows.

%For the sake of the following discussion, we will also say $v$ \emph{covers} some (but not all) partial trees generated by $\aut$ c set of states at the bottommost level of some $\aut$'s run is exactly $D$.

%According to the aforementioned intuition, $D$ actually contains the set of states at some level of a run $\run$ of $\aut$, and every $f_i$ contains the set of states at the same level of some possible run $\rho_i$ of $\autb$ on the same tree and how it overlaps with the run $\rho$ of $\aut$. More specifically, $q_\autb\in f_i(q_\aut)$ for some $q_\aut\in D$ iff some node of a tree corresponding to $q_\aut$ in $\rho$ also corresponds to $q_\autb$ in $\rho_i$.
%For instance, in~\cref{fig:inclusionExample}, the state $q$ of $\aut$ corresponds to the states $s,r$ of $\autb$ because they are used in the same tree level and the same tree nodes. So we have $f_i(q) = \{r,s\}$. %Each vertex represents \emph{some} trees generated by $\aut$ whose . The set $D$ contains all bottoms of all used transitions from the previous vertex to the current vertex. \wl{For each aforementioned tree, all possible $f_i$'s such that each bottommost node in that tree corresponding to $q_A\in Q_\aut$ also corresponds to $q_B\in f_i(q_A)\subseteq Q_\autb$ will appear in $\{f_1,\ldots,f_m\}$ if the tree can be generated by $\autb$ simultaneously.}

First, we compute possible successors of the~$D$-component of~$v$.
To do this, we need to explore all \emph{feasible} sets $\Gamma_\aut$ of
transitions from~$D$ in~$\aut$.
More concretely, in each set~$\Gamma_\aut$, we select exactly one downward
transition $\delta_{q_\aut}$ originating from each $q_\aut \in D$, such that
all transitions in $\Gamma_\aut$ share a~common choice (as required by the
definition of an accepting run (cf. \cref{sec:lsta}).
Formally, given $D = \{q_1, \ldots, q_k)$, we consider all sets of
transitions~$\Gamma_\aut = \{\delta_1, \ldots, \delta_k\}$ such that the following formula holds:
\begin{equation}
  % \left(\forall 1\le i\le k\colon \delta_i\in\Delta \land \topof{\delta_i}=q_i\right) \land \bigcap_{1\le i\le k}\ell(\delta_i)\ne\emptyset
  \left(\forall 1\le i\le k\colon \delta_i\in\Delta \land \topof{\delta_i}=q_i\right)\quad \land \quad \textstyle\bigcap\{\ell(\delta_i) \mid 1\le i\le k\}\ne\emptyset.
\end{equation}
We will denote the set of all such $\Gamma_\aut$'s from a~set of states~$D$ as $\feasibletransof \aut D$.

\newcommand{\algStateCorr}[0]{
\begin{algorithm}[t]
\KwIn{Set of transitions $\Gamma_\aut\subseteq \Delta_\aut$, map $f\colon\{\topof{\delta} \mid \delta\in\Gamma_\aut\} \to 2^{Q_\autb}$}
\KwOut{The set of pairs $(f', g')$ of mappings compatible with~$\Gamma_\aut$ obtained from~$f$}
\SetArgSty{normalfont}
\SetKw{KwBreak}{break}
\SetKwProg{Func}{Procedure}{:}{}
\SetKwIF{IfNot}{ElseIfNot}{}{if not}{then}{else if not}{}{}
\SetKwFunction{FindAllFeasibleTransitions}{FindAllFeasibleTransitions}
\SetKwFunction{FindAllMappings}{FindAllMappings}
% $Q \gets \bigcup\left\{f(q_\aut) \mid q_\aut\in\domof f \right\}$;
$Q \gets \bigcup \imgof f$;\label{ln:imgoff}
$F' \gets \emptyset$\;
\ForEach{$\Gamma_\autb\in\ \feasibletransof \autb Q$}{\label{ln:feasibleb}
    $f' \gets \emptyset$;
    $g' \gets \emptyset$;
    $\mathit{failed} \gets \FF$\;
    % $\mathit{pairs} \gets \emptyset$;
    % $\mathit{success} \gets \TT$\;
    % \ForEach{$\delta_\aut\in\Gamma_\aut$ and $\delta_\autb\in \Gamma_\autb$ such that $\topof{\delta_\autb}\in f(\topof{\delta_\aut})$}{\label{ln:forgammastart}
    \ForEach{$\delta_\aut\in\Gamma_\aut$ and $q \in f(\topof{\delta_\aut})$}{\label{ln:forgammastart}
      $\{\delta_\autb\} \gets \{\delta_\autb \in \Gamma_\autb \mid q = \topof{\delta_\autb}\}$\label{ln:find:deltaB}\;
      % \lIf{$B = \emptyset$}{
      %   $\mathit{failed} \gets \TT$; \KwBreak
      % }
      % $\{\delta_B\} \gets B$\;
      \lIf{$a = \symof{\delta_\aut}$ and $b=\symof{\delta_\autb}$ are both leaf symbols}{\label{ln:symbolcheckstart}\label{ln:symbolsleaf}
        % $B \gets \{\delta_\autb\in \Gamma_\autb \mid \topof{\delta_\autb}\in
        %   f(\topof{\delta_\aut}), \symof{\delta_\autb} \text{ is a leaf
          % symbol}\}$\;
      $g'(a).\text{insert}(b)$
      }
      \ElseIf{$\symof{\delta_\aut} = \symof{\delta_\autb}$ is an internal symbol}{\label{ln:symbolsinternal}
          $f'(\leftof{\delta_\aut})$.insert($\leftof{\delta_\autb}$);
          $f'(\rightof{\delta_\aut})$.insert($\rightof{\delta_\autb}$)\;
      }\label{ln:symbolcheckend}
      \lElse{\label{ln:symbolsother}
        $\mathit{failed} \gets \TT$; \KwBreak
      }
      \label{ln:forgammaend}
  }
  \lIf{$\neg \mathit{failed}$}{
    $F'.\text{insert}((f', g'))$
  }
}
\KwRet $F'$\;\label{ln:mappingreturn}
\caption{FindAllMappings($\Gamma_\aut, f$).}
\label{algo:statecorr}
\end{algorithm}
}

\newcommand{\algInclusion}[0]{
\begin{algorithm}[t]
\KwIn{\lstas $\aut = \tuple{Q_\aut, \mathbb N\cup \symterm, \Delta_\aut, \rootstates_\aut, \globconstr_\aut}$, $\autb=\tuple{Q_\autb, \mathbb N\cup \symterm, \Delta_\autb, \rootstates_\autb = \{r_1,\ldots,r_k\}, \globconstr_\autb}$.}
\KwOut{$\TT$ if $\aut \entailedbyuptosc \autb$, $\FF$ otherwise}
% \KwOut{$\TT$ if $\forall\ T_\aut\in\lang(\aut),\ \exists\ T_\autb\in\lang(\autb),\ \forall\mathbb X,\ \globconstr_\aut \land \globconstr_\autb \implies \exists\ r\in\mathbb{R}-\{0\}\ s.t.\ T_\aut = r\cdot T_\autb$; otherwise, $\FF$.}
\SetArgSty{normalfont}
\SetKw{KwBreak}{break}
\SetKwProg{Func}{Procedure}{:}{}
\SetKwIF{IfNot}{ElseIfNot}{}{if not}{then}{else if not}{}{}

% \SetKwFunction{CreateSourceVertices}{CreateSourceVertices}
% \Func{\CreateSourceVertices{}}{
%     $queue \gets$ Queue()\;
%     \ForEach(\tcp*[f]{the order does not matter}){$q\in\rootstates_\aut$}{
%         $queue.push((\{q\},\{\{q\mapsto \{r_1\}\}, \{q\mapsto \{r_2\}\},\ldots,\{q\mapsto \{r_k\}\}\}))$
%     }
%     \KwRet $queue$\;
% }
\SetKwFunction{FindAllFeasibleTransitions}{FindAllFeasibleTransitions}
\SetKwFunction{FindAllMappings}{FindAllMappings}
$\mathit{processed} \gets \emptyset$\;
% $\mathit{workset} \gets \{(\{q\},\ \{\{q\mapsto \{r_1\}\},\ \ldots,\ \{q\mapsto \{r_k\}\}\}) \mid q\in\rootstates_\aut\}$\;\label{ln:sourcevertices} %\CreateSourceVertices{}\;
$\mathit{workset} \gets \{(\{q\},\ \{(\{q\mapsto \{r_1\}\}, \emptyset), \ldots, (\{q\mapsto \{r_k\}\}, \emptyset)\}) \mid q\in\rootstates_\aut\}$\;\label{ln:sourcevertices} %\CreateSourceVertices{}\;
\While{$\exists (D,F) \in \mathit{workset}$}{\label{ln:entail:nonempty}
    $\mathit{workset}.\text{remove}((D,F))$\;\label{ln:entail:remove}
    $\mathit{processed}.\text{insert}((D,F))$\;
    \If{$D=\emptyset$}{\label{ln:entail:terminalstart}
      $\varsY \gets \{\varsof{u_1} \mid (u_1 \mapsto U_2) \in g \land (\emptyset, g) \in F\}\cup\varsof{\globconstr_\aut}$\label{ln:entail:checkstart}\;
      $\varsZ \gets \{\varsof{U_2} \mid (u_1 \mapsto U_2) \in g \land (\emptyset, g) \in F\}\cup\varsof{\globconstr_\autb} \setminus \varsY$\;
      \If{$\displaystyle\neg\forall\varsY \colon \globconstr_{\aut} \implies
         \exists\varsZ\colon \globconstr_{\autb} \land \bigvee_{\mathclap{(\emptyset, g)
         \in F}} \exists r\in\reals\setminus\{0\}\colon
        \bigwedge_{\mathclap{(u_1 \mapsto U_2) \in g\atop u_2 \in U_2}} u_1 = r \cdot u_2$}{\label{ln:entail:check}
        \vspace{-5mm}
        \Return{$\FF$}\;\label{ln:entail:terminalend}\label{ln:entail:checkend}
      }
      % \vspace{-1mm}
      % \ForEach{$(\emptyset, g) \in F$}{
      %   $\varsY \gets \{\varsof{u_1} \mid (u_1 \mapsto U_2) \in g\}$\label{ln:entail:checkstart}\;
      %   $\varsZ \gets \{\varsof{U_2} \mid (u_1 \mapsto U_2) \in g\} \setminus \varsY$\;
      %   % \vspace{-3mm}
      %   \If{$\displaystyle\forall\varsY \colon \globconstr_{\aut} \implies
      %      \exists\varsZ\colon \globconstr_{\autb} \land \exists r\in\reals\setminus\{0\}\colon
      %     \bigwedge_{\mathclap{(u_1 \mapsto U_2) \in g\atop u_2 \in U_2}} u_1 = r \cdot u_2$}{
      %       \label{ln:entail:checkend}
      %     \vspace{-5mm}
      %     $\mathit{covered} \gets \TT$; \KwBreak\;
      %   }
      %   \vspace{-1mm}
      % }
      % \lIf{$\neg\mathit{covered}$}{
      %   \Return $\FF$\label{ln:entail:terminal}\label{ln:entail:terminalend}
      % }
    }
    \ForEach{$\Gamma_\aut\in\feasibletransof \aut D$}{\label{ln:entail:foreach}
        $D' \gets \big\{q\in\botof{\delta} \mid \delta\in\Gamma_\aut\big\}$; $F' \gets \emptyset$\;\label{ln:entail:d}
        \ForEach{$(f,g) \in F$}{\label{ln:entail:F}
            $U\gets {}$\FindAllMappings{$\Gamma_\aut, f$}\;
            $F' \gets F' \cup \{(f', g \merge g') \mid (f', g') \in U\}$\label{ln:entail:merge}\;
        }
        \lIf{$(D', F')\not\in \mathit{processed} \cup \mathit{workset}$}{
          $\mathit{workset}.\text{insert}((D', F'))$\label{ln:entail:insert}
        }
    }
}
\Return \TT\;
\caption{Checking if $\aut \entailedbyuptosc \autb$.}
\label{algo:inclusion}
\end{algorithm}
}

\algStateCorr

Next, we will show how to construct the set of all feasible pairs of mappings
$\{(f'_1, g'_1),
\ldots, (f'_n, g'_n)\}$ for some $\Gamma_\aut$ and a~pair of upper level mappings~$(f,g)$.
Each pair of mappings $(f'_j, g'_j)$ at the lower level is derived from some pair of mappings~$(f_i,g_i)$ at
the upper level and a set of downward transitions $\Gamma_\autb$ from~$\autb$.
The construction process is described in~\cref{algo:statecorr}, where we use
$(f, g)$ to denote an upper level state and term mapping respectively and
$(f',g')$ to denote their lower level counterparts.

The basic idea of the algorithm is simple:
\begin{inparaenum}[(1)]
  \item  we use the upper level mapping $f$ and transitions
    $\Gamma_\aut$ to compute the set of top states $Q$ (\lnref{ln:imgoff}), 
  \item  then, we find all feasible transitions $\Gamma_\autb$ from $Q$
    (\lnref{ln:feasibleb}), and, finally,
  \item  for each pair $(\Gamma_\aut, \Gamma_\autb)$, we construct one pair of lower
    level state and term mappings~$(f',g')$ (\lnrangeref{ln:forgammastart}{ln:forgammaend}).
    \label{item:one-lower}
\end{inparaenum}
Specifically, in step (\ref{item:one-lower}),
it must hold that for each transition~$\delta_\aut \in \Gamma_\aut$,
every state $q \in f(\topof{\delta_\aut})$ needs to be able to
match~$\delta_\aut$ by a~transition from~$\Gamma_\autb$.
To check this, for every such~$q$ we select from~$\Gamma_\autb$ the transition
$\delta_\autb$, which is the transition of~$\Gamma_\autb$ with~$q$ as its top
(it follows from $\feasibletrans$ that there is exactly one).
There are three possible cases:
\begin{itemize}
    \item \emph{Both $\delta_{\aut}$ and $\delta_{\autb}$ are leaf transitions
      (\lnref{ln:symbolsleaf}):} we remember in~$g'$ that the symbol of $\delta_\autb$
      need to be able to match the symbol of $\delta_\aut$, which will be
      checked later for all such matchings together.
    \item \emph{The symbols of both $\delta_{\aut}$ and $\delta_{\autb}$ are
      internal (\lnref{ln:symbolsinternal}):} in this case, we add new entries
      to the lower level mapping~$f'$.
    \item \emph{One transition is internal while the other is leaf
      (\lnref{ln:symbolsother}):} the pair $(\Gamma_\aut, \Gamma_\autb)$
      cannot form a feasible lower level mapping.
\end{itemize}
%
% For the first two cases, we found a potential feasible lower level mapping $f'$
% (\lnref{ln:success}).
% For the case of leaf transitions (which are labelled by terms over variables
% from~$\vars$ and complex numbers from~$\complex$), we, however, need to further
% verify if they are equivalent up to scaling (\lnref{ln:utscheck}).
% Specifically, we need to ensure for all assignments to $\vars$ that satisfy the
% common global constraints $\globconstr_\aut \land \globconstr_\autb$, the leaf
% symbols of $\aut$, i.e., $t_1$, are a scaled version of the corresponding leaf
% symbols of $\autb$, i.e., $t_2$. 
% \ol{return here later}

\algInclusion

Finally, the main entailment testing routine is summarized in~\cref{algo:inclusion}.
\lnref{ln:sourcevertices} creates the set of source vertices of the explored graph.
\lnrangeref{ln:entail:nonempty}{ln:entail:remove} pick a~vertex $(D,F)$ that has
not been processed yet.
\lnrangeref{ln:entail:terminalstart}{ln:entail:terminalend} check whether
$(D,F)$ is a terminal vertex and conclude that the entailment test fails when
such a~vertex is reached.
This check consists of looking at all pairs $(\emptyset, g)$ in~$F$ and checking
whether there is a~way how the $\autb$-terms from the~$g$'s can together cover
(modulo a~scaling factor) the behaviour of the $\aut$-terms.
\lnrangeref{ln:entail:foreach}{ln:entail:insert} are the edge construction
procedure.
\lnrangeref{ln:entail:foreach}{ln:entail:d} enumerate all feasible
$\Gamma_\aut$ and use them to create the next vertex $(D',F')$.
Specifically, $D'$ are the bottom states from $\Gamma_\aut$
(\lnref{ln:entail:d}), and $F'$ are the union of all feasible mappings of
$\Gamma_\aut$ (\lnrangeref{ln:entail:F}{ln:entail:insert}). 
The successor term mapping is computed from the upper-level one and from the
one returned by \FindAllMappings by, for each term of~$\aut$, merging the
corresponding sets of terms of the two mappings: $g \merge g' = \{x \mapsto Y
\mid x \in \domof{g \cup g'}, Y = g(x) \cup g'(x)\}$ (\lnref{ln:entail:merge}).

A~crucial part of the algorithm is the term mapping~$g$ check between
the terms from~$\aut$ and the terms from~$\autb$
(\lnrangeref{ln:entail:checkstart}{ln:entail:checkend}).
Here, we want to check that for every possible value that we can obtain from
a~term~$t$ on the left-hand side of the entailment (satisfying~$\aut$'s global
constraint~$\globconstr_\aut$) and every term~$t_i$ from $g(t) = \{t_1, \ldots,
t_k\}$, there is a way to obtain a~value (satisfying $\autb$'s global
constraint~$\globconstr_\autb$) such that for all these values, there is
a~common scaling factor~$r$ to make them equal to the value from~$t$.
We emphasize the way how we need to deal with the quantified variables.
Variables $\varsY$ occurring on the left-hand side (and possibly also on the
right-hand side, since we may need to synchronize the values) are quantified
universally, while variable~$\varsZ$ occuring only on the right-hand side of the
entailment are quantified existentially.
This scaling check requires invoking an SMT solver with a formula in the \texttt{NIRA}
(non-linear integer and real arithmetic) logic.
In the special case where all leaf symbols are constants, $r$~is the only
variable, and global constraints are $\TT$, the problem reduces to a simple \texttt{QF\_LRA} (quantifier-free
linear real arithmetic) formula.
\begin{theorem}[Soundness]
    When \cref{algo:inclusion} terminates, it returns true iff $\aut \entailedbyuptosc \autb$.
\end{theorem}

\begin{theorem}[Termination]
  % \cref{algo:inclusion} generates a graph with at most $2^{|Q_\aut| +
  % 2^{|Q_\aut|\cdot|Q_\autb|}}$ vertices. When the leaf symbols
  % of either $\aut$ or $\autb$ are only constants, the algorithm always
  % terminates.
  When the terms in leaf symbols and the global constraints of $\aut$
  and~$\autb$ use a~decidable theory, the algorithm always terminates.
\end{theorem}
\begin{proof}
Since the number of states and terms occurring in~$\aut$ and~$\autb$ is finite, 
the constructed graph is also finite.
Further, since the underlying theory for the terms and global constraints is assumed to be
decidable, the check at \lnref{ln:entail:check} always terminates.
% Since every node of the graph has the structure $(D, \{f_1, \ldots, f_m\})$,
% the set of vertices is included in the set $2^{Q_\aut} \times 2^{Q_\aut \to
% 2^{Q_\autb}} = 2^{Q_\aut} \times 2^{(2^{Q_\autb})^{Q_\aut}}$.
% % We can view $f_i$ as a mapping of the form $f_i\colon Q_\aut \to 2^{Q_\autb}$.
% This allows us to bound the number of vertices in the graph by the following expression:
% %
% \begin{equation}
%   2^{|Q_\aut|} \cdot 2^{(2^{|Q_\autb|})^{|Q_\aut|}}
%   =
%   2^{|Q_\aut|} \cdot 2^{2^{|Q_\autb|\cdot|Q_\aut|}}
%   =
%   2^{|Q_\aut| + 2^{|Q_\aut|\cdot|Q_\autb|}}.
% \end{equation}
% %
% The graph construction involves solving SMT formulae. In the general case, we need to handle sub-formulae of the form $t_1 = r \cdot t_2$ (or equivalently $t_2 = r \cdot t_1$), which require solving non-linear constraints from the \texttt{NIRA} category.
%   Such SMT queries do not guarantee termination.
%   However, when the leaf symbols of either $\aut$ or $\autb$ are constants (i.e., when either $t_1$ or $t_2$ is constant), the satisfiability of the formula can be determined.
\qed
\end{proof}
\newcommand{\EXPSPACE}[0]{\textsf{\textbf{{EXPSPACE}}}}

% We note that although the worst-case complexity of checking the entailment relation is
% infeasible (double exponential), it was never a bottleneck in our experiments.
% \ol{verify}

%%%%%%%%%%%%%%%%%%%%%%%%%%%%%%%%%%%%%%%%%%%%%%%%%%%%%%%%%%%%%%%%%%%%%%%%%%%%%%%%
\vspace{-3.0mm}
\section{Experimental Results}\label{sec:experiments}
\vspace{-2.0mm}
%%%%%%%%%%%%%%%%%%%%%%%%%%%%%%%%%%%%%%%%%%%%%%%%%%%%%%%%%%%%%%%%%%%%%%%%%%%%%%%%

\newcommand{\tabExp}[0]{
\begin{table}[t] \centering
    \caption{Results of verifying some real-world examples with \autoqq}
\vspace{0.0cm}
\label{tab:exp}
\scalebox{1}{\scalebox{0.82}{
\begin{tabular}{lrrccrrclrrccrr}\toprule
\multicolumn{7}{c}{\emph{Weakly Measured Grover's Search}~\cite{weakly:chris}} &&
\multicolumn{7}{c}{\emph{Repeat-Until-Success}~\cite{RUS}}\\
\cmidrule{1-7}
\cmidrule{9-15}
\multicolumn{1}{c}{\textbf{program}} &
\multicolumn{1}{c}{\textbf{qubits}} &
\multicolumn{1}{c}{\textbf{gates}} &
~~ &
\multicolumn{1}{c}{\textbf{result}} &
\textbf{time} &
\multicolumn{1}{c}{\textbf{memory}} &
~~ &
\multicolumn{1}{c}{\textbf{program}} &
\multicolumn{1}{c}{\textbf{qubits}} &
\multicolumn{1}{c}{\textbf{gates}} &
~~ &
\multicolumn{1}{c}{\textbf{result}} &
\textbf{time} &
\multicolumn{1}{c}{\textbf{memory}}\\
\cmidrule{1-7}
\cmidrule{9-15}
WMGrover (03) & 7 & 50 && OK & 0.0s & 42MB && $(2X+\sqrt2 Y+Z)/\sqrt7$ & 2 & 29 && OK & 0.0s & 7MB\\
WMGrover (10) & 21 & 169 && OK & 0.2s & 42MB && $(I+i\sqrt2X)/\sqrt3$ & 2 & 17 && OK & 0.0s & 7MB\\
WMGrover (20) & 41 & 339 && OK & 0.8s & 42MB && $(I+2iZ)/\sqrt5$ & 2 & 27 && OK & 0.0s & 6MB\\
WMGrover (30) & 61 & 509 && OK & 2.3s & 43MB && $(3I+2iZ)/\sqrt13$ & 2 & 43 && OK & 0.0s & 7MB\\
WMGrover (40) & 81 & 679 && OK & 5.4s & 43MB && $(4I+iZ)/\sqrt17$ & 2 & 77 && OK & 0.0s & 6MB\\
WMGrover (50) & 101 & 849 && OK & 11s & 44MB && $(5I+2iZ)/\sqrt29$ & 2 & 69 && OK & 0.0s & 7MB\\
\bottomrule
\end{tabular}
}}
\vspace*{-7mm}
\end{table}
}

We demonstrate the use of \autoqq~\cite{autoqrepo} on two real-world use cases
consisting of quantum programs with loops that were proposed
in~\cite{RUS,weakly:chris}.
We ran all experiments on a server running Ubuntu 22.04.3 LTS with an
AMD EPYC 7742 64-core processor (1.5\,GHz), 1,152\,GiB of RAM, and a 1\,TB SSD.
% The results are given in \cref{tab:exp}.
% \vspace{-30pt}
% The tool used in this paper is
% currently available at~\cite{autoqrepo}. \ol{fix the authors?} %\url{https://github.com/alan23273850/AutoQ/tree/CAV24}. %\wl{is this citation ok?}

%%%%%%%%%%%%%%%%%%%%%%%%%%%%%%%%%%%%%%%%%%%%%%%%%%%%%%%%%%%%%%%%%%%%%%
\newcommand{\groveralg}[0]{
\begin{wrapfigure}[11]{r}{6.8cm}
\vspace*{-7mm}
\hspace*{-3mm}
\scalebox{0.9}{
\begin{minipage}{7.7cm}
\begin{algorithm}[H]
    \caption{A Weakly Measured Version of Grover's algorithm (solution $s=0^n$)}
    \label{alg:Grover}

    {\color{gray} Pre: $\{\specpair{0^{n+2}}{1}+ \specdef{0}\}$\;} % \otimes(\specpair{0}{-a_1}, \specdef{-a_0})\}
    % $H_1;\ H_2;\ \ldots;\ H_{n+2}$\label{ln:walshhadamard}\;
    $H_3;\ H_4;\ \ldots;\ H_{n+2}$\label{ln:walshhadamard}\;
    $\mathcal{O}_{2,\ldots,(n+2)};\mathit{CK}^{2}_{1};\mathcal{O}_{2,\ldots,(n+2)}$\;
    {\color{gray} Inv:
    $\{\specpair{000^{n}}{v_{\mathit{sol1}}}+\specpair{000^{n-1}1}{v_k}+\cdots+{}$\\
    \hspace*{8mm}$\specpair{001^{n}}{v_k}+\specpair{100^{n}}{v_{\mathit{sol2}}}+\specdef{0}\}$\;} %
    \While{$M_1=0$}{
      \{$\mathcal{G}_{2,\ldots,(n+2)};\mathcal{O}_{2,\ldots,(n+2)};\mathit{CK}^{2}_{1};\mathcal{O}_{2,\ldots,(n+2)}$\}\;
    }
    {\color{gray} Post: $\{\specpair{10s}{1} + \specdef{0}\}$\;}
\end{algorithm}
\end{minipage}
}
\end{wrapfigure}
}

%%%%%%%%%%%%%%%%%%%%%%%%%%%%%%%%%%%%%%%%%%%%%%%%%%%%%%%%%%%%%%%%%%%%%%
\newcommand{\groverfigwrap}[0]{
\begin{wrapfigure}[19]{r}{58mm}
 \vspace{-15mm}
\hspace*{-3mm}
\scalebox{0.782}{
% \resizebox{1.04\textwidth}{!}{
\begin{tabular}{cc}
\vspace{-3mm}
\subcaptionbox{Uniform below $\ket{00}$\label{fig:grovera}}{\scalebox{.7}{\tA}} &
\subcaptionbox{Applied $\mathcal{O}_{2,\ldots,(n+2)}$\label{fig:groverb}}{\scalebox{.7}{\tB}}\\\\
\vspace{-3mm}
  \subcaptionbox{Applied $\mathit{CK}^2_1$\label{fig:groverc}}{\scalebox{.7}{\tC}}&
\subcaptionbox{Applied $\mathcal{O}_{2,\ldots,(n+2)}$ again\label{fig:groverd}}{\scalebox{.7}{\tD}}\\\\
\subcaptionbox{When $x_1$ is measured $0$\label{fig:grovere}}{\scalebox{.7}{\tE}} &
\subcaptionbox{Applied $\mathcal{G}_{2,\ldots,(n+2)}$\label{fig:groverf}}{\scalebox{.7}{\tF}}
\end{tabular}}
% \tG \tH
\caption{Intermediate states of~\cref{alg:Grover}}

\label{fig:evolution}
\end{wrapfigure}
}

%%%%%%%%%%%%%%%%%%%%%%%%%%%%%%%%%%%%%%%%%%%%%%%%%%%%%%%%%%%%%%%%%%%%%%
\newcommand{\groverfig}[0]{
\begin{figure}[t]
\begin{subfigure}[b]{0.33\linewidth}
\scalebox{.7}{\tA}
\vspace{-1mm}
\caption{Uniform below $\ket{00}$}
\label{fig:grovera}
\end{subfigure}
\begin{subfigure}[b]{0.33\linewidth}
\scalebox{.7}{\tB}
\vspace{-1mm}
\caption{Applied $\mathcal{O}_{2,\ldots,(n+2)}$}
\label{fig:groverb}
\end{subfigure}
\begin{subfigure}[b]{0.33\linewidth}
\scalebox{.7}{\tC}
\vspace{-1mm}
\caption{Applied $\mathit{CK}^2_1$}
\label{fig:groverc}
\end{subfigure}

\begin{subfigure}[b]{0.33\linewidth}
\scalebox{.7}{\tD}
\vspace{-1mm}
\caption{Applied $\mathcal{O}_{2,\ldots,(n+2)}$ again}
\label{fig:groverd}
\end{subfigure}
\begin{subfigure}[b]{0.33\linewidth}
\scalebox{.7}{\tE}
\vspace{-1mm}
\caption{When $x_1$ is measured $0$}
\label{fig:grovere}
\end{subfigure}
\begin{subfigure}[b]{0.33\linewidth}
\scalebox{.7}{\tF}
\vspace{-1mm}
\caption{Applied $\mathcal{G}_{2,\ldots,(n+2)}$}
\label{fig:groverf}
\end{subfigure}
\vspace{-5mm}
\caption{Intermediate states of~\cref{alg:Grover}}
\label{fig:evolution}
\vspace*{-5mm}
\end{figure}
}

%%%%%!!!!!!!!!!!!!!!!!!!!!!!!
\pagebreak 
%%%%%!!!!!!!!!!!!!!!!!!!!!!!!

%*******************************************************************************
\vspace{-3.0mm}
\subsection{The Weakly Measured Version of Grover's Algorithm}
\vspace{-1.0mm}
%*******************************************************************************
\groveralg
\emph{Grover's algorithm}~\cite{Grover96}, introduced in 1996, is a quantum algorithm
that performs unstructured search.
Given an \emph{oracle function} (which can say whether a~particular
%
% \linebreak
% %!!!!!!!!!!!!!!!!!!!!!
% % BROKEN PARAGRAPH
% % !!!!!!!!!!!!!!!!!!!!
%
% \groveralg
% \noindent
binary
assignment is a~solution), Grover's algorithm can efficiently find a solution
(with high probability).
The algorithm requires approximately $\mathcal{O}(\sqrt{N/k})$ evaluations of
the oracle function, where $N$ is the size of the function's domain
(usually~$2^n$ for~$n$ qubits), and $k$ is the number of solutions.
The number of solutions is, however, not always known, making it difficult to
determine the algorithm's parameters (the algorithm is sensitive to the number
of evaluations; in particular, doing more evaluations may make the
probability of finding the solution smaller).
To address this issue, a~variation of Grover's search, called the \emph{weakly
measured} version (cf.~\cref{alg:Grover}), was recently proposed~\cite{weakly:chris}.
The weakly measured version eliminates the need for knowing the number of solutions, making
the algorithm more applicable.

% \groverfigwrap

\groverfig

To explain the algorithm, we first introduce some of its key components.
% The \emph{oracle} circuit, denoted as $\mathcal{O}_{i, \ldots, j}$, works from
% qubits $x_i$ to $x_j$, where $x_i$ is the ancilla qubit, and $x_{i+1}$ to $x_j$
% are the working qubits.
The algorithm works over qubits $x_1, \ldots, x_{n+2}$.
\lnref{ln:walshhadamard} first applies the \emph{Walsh–Hadamard} transform to
obtain the superposition on all qubits other than~$x_1$ and~$x_2$ (which are two
ancillas), obtaining the state in \cref{fig:grovera}.
The \emph{oracle} circuit, denoted as $\mathcal{O}_{2, \ldots, (n+2)}$, works from
qubits $x_2$ to $x_{n+2}$, where $x_2$ is the ancilla qubit and $x_3$ to
$x_{n+2}$ are the working qubits.
As shown from~\cref{fig:grovera,fig:groverb} (and also
from~\cref{fig:groverc,fig:groverd}), the effect of the oracle circuit
is to flip the ancilla qubit of the computational bases corresponding to the
solutions.
That is, it swaps the amplitude values of $\ket{0s}$ and $\ket{1s}$, for all
solutions~$s$.
The oracle circuit can be constructed using gates supported in $\autoqq$. %E.g., when the solution is $1^n$, ancilla qubit is $x_1$, working qubits are $x_2\ldots x_{n+1}$, we have $\mathcal{O}_{1, \ldots, n+1}= MCX^{2\ldots n+1}_1$, where $MCX^{2\ldots n+1}_1$ is the multi-controlled $X$ gate with the control qubits $x_2\ldots, x_{n+1}$ and target qubit $x_1$.

The controlled rotation gate ${CK}^i_j$ is a~special gate supported
in \autoqq.
% , where $x_i$ and $x_j$ are the controlled and target qubits,
% respectively.
In this algorithm, the gate always applies to a target qubit whose value is
$\ket{0}$ when the controlled qubit is $\ket{1}$, and it updates the target
qubit to $a\ket{1}+b\ket{0}$ with $a^2+b^2=1$, for some very small $b$.
In \autoqq, we use $a=\frac{21}{221}$ and $b=\frac{220}{221}$.
We demonstrate the behavior of ${CK}^2_1$ from \cref{fig:groverb,fig:groverc},
where a~small portion ($\frac{21^2}{221^2} \approx 1\,\%$) of the
probability under the branch~$\ket{01}$ is moved to the branch~$\ket{11}$, as shown
in~\cref{fig:groverc}.
After applying $\mathcal{O}_{2, \ldots, (n+2)}$ again, we obtain the state in
\cref{fig:groverd} (this state is captured by the loop invariant).
Here, we can already measure the qubit~$x_1$ and if the result is~1, this
collapses the probability of the left sub-tree of~$x_1$ in \cref{fig:groverd}
to~0, so the only non-zero probability basis is the solution
$\ket{10s}$.

Otherwise (the result of measuring~$x_1$ was 0), we enter the loop, which contains the \emph{Grover iteration}
circuit, denoted as $\mathcal{G}_{2, \ldots, (n+2)}$, which also uses
$\mathcal{O}_{2, \ldots, (n+2)}$ as a~component.
The effect of $\mathcal{G}_{2, \ldots, (n+2)}$
is to increase the probabilities of basis states for the solutions
and decrease others, as shown in~\cref{fig:grovere,fig:groverf}.
After $\mathcal{G}_{2, \ldots, (n+2)}$, we execute the same sequence
\emph{oracle-rotation-oracle} as above to obtain a~state resembling
\cref{fig:groverd}.
We keep repeating until we measure~$x_1 = 0$, in which case we terminate
with a~solution.

\begin{table}[t] \centering
    \caption{Results of verifying some real-world examples with \autoqq}
\vspace{0.0cm}
\label{tab:exp}
\scalebox{1}{\scalebox{0.82}{
\begin{tabular}{lrrccrrclrrccrr}\toprule
\multicolumn{7}{c}{\emph{Weakly Measured Grover's Search}~\cite{weakly:chris}} &&
\multicolumn{7}{c}{\emph{Repeat-Until-Success}~\cite{RUS}}\\
\cmidrule{1-7}
\cmidrule{9-15}
\multicolumn{1}{c}{\textbf{program}} &
\multicolumn{1}{c}{\textbf{qubits}} &
\multicolumn{1}{c}{\textbf{gates}} &
~~ &
\multicolumn{1}{c}{\textbf{result}} &
\textbf{time} &
\multicolumn{1}{c}{\textbf{memory}} &
~~ &
\multicolumn{1}{c}{\textbf{program}} &
\multicolumn{1}{c}{\textbf{qubits}} &
\multicolumn{1}{c}{\textbf{gates}} &
~~ &
\multicolumn{1}{c}{\textbf{result}} &
\textbf{time} &
\multicolumn{1}{c}{\textbf{memory}}\\
\cmidrule{1-7}
\cmidrule{9-15}
WMGrover (03) & 7 & 50 && OK & 0.0s & 42MB && $(2X+\sqrt2 Y+Z)/\sqrt7$ & 2 & 29 && OK & 0.0s & 7MB\\
WMGrover (10) & 21 & 169 && OK & 0.2s & 42MB && $(I+i\sqrt2X)/\sqrt3$ & 2 & 17 && OK & 0.0s & 7MB\\
WMGrover (20) & 41 & 339 && OK & 0.8s & 42MB && $(I+2iZ)/\sqrt5$ & 2 & 27 && OK & 0.0s & 6MB\\
WMGrover (30) & 61 & 509 && OK & 2.3s & 43MB && $(3I+2iZ)/\sqrt13$ & 2 & 43 && OK & 0.0s & 7MB\\
WMGrover (40) & 81 & 679 && OK & 5.4s & 43MB && $(4I+iZ)/\sqrt17$ & 2 & 77 && OK & 0.0s & 6MB\\
WMGrover (50) & 101 & 849 && OK & 11s & 44MB && $(5I+2iZ)/\sqrt29$ & 2 & 69 && OK & 0.0s & 7MB\\
\bottomrule
\end{tabular}
}}
\vspace*{-7mm}
\end{table}
   %%%%%%%

The results of the verification of weakly measured Grover's search are in the
left-hand side of \cref{tab:exp}: \autoqq was able to verify the program
w.r.t.\ the specification even for larger numbers of qubits in reasonable time.

\vspace{-3.0mm}
\subsection{Unitaries as Repeat-Until-Success Circuits}
\vspace{-2.0mm}
%*******************************************************************************
% \enlargethispage{\baselineskip}
% \enlargethispage{\baselineskip}
% \enlargethispage{\baselineskip}

%!!!!!!!!!!!!!!!!!!!!!!!!!!!!!!!!
\enlargethispage{3mm}
%!!!!!!!!!!!!!!!!!!!!!!!!!!!!!!!!

\emph{Repeat-until-success} programs are a~general framework that was introduced
to simplify quantum circuit decomposition (we introduced an example of generating the ``$-X$'' gate via
the RUS framework in~\cref{sec:overview}).
RUS programs have been shown to be more efficient (in terms of circuit depth)
than ancilla-free techniques when it comes to synthesizing single-qubit gates
(cf.~\cite{RUS,RUS2}).
We present the results of verification of RUS programs for generating various
non-standard gates in the right-hand side of \cref{tab:exp}.
Note that \autoqq can verify these programs instantaneously.

%%%%%%%%%%%%%%%%%%%%%%%%%%%%%%%%%%%%%%%%%%%%%%%%%%%%%%%%%%%%%%%%%%%%%%%%%%%%%%%%
\vspace{-3.0mm}
\section{Conclusion and Future Work}\label{sec:label}
\vspace{-2.0mm}
%%%%%%%%%%%%%%%%%%%%%%%%%%%%%%%%%%%%%%%%%%%%%%%%%%%%%%%%%%%%%%%%%%%%%%%%%%%%%%%%

We presented a~major extension of~\autoq~\cite{chen2023autoq} with an added support for control flow
constructs and evaluated its feasibility on a~family of programs for
the weak-measurement-based version of Grover's algorithm and on implementations
of a~number of non-standard quantum gates using repeat-until-success circuits.
In the future, we wish to extend the framework with automating invariant
generation (e.g., using a~modification of the symbolic-execution-based technique
from~\cite{ChenCJJL24}) and add support for dealing with more complex loops
that give rise to mixed states.

\pagebreak 

%%%%%%%%%%%%%%%%%%%%%%%%%%%%%%%%%%%%%%%%%%%%%%%%%%%%%%%%%%%%%%%%%%
\bibliographystyle{splncs04}
\bibliography{literature}

\begin{thebibliography}{10}
\providecommand{\url}[1]{\texttt{#1}}
\providecommand{\urlprefix}{URL }
\providecommand{\doi}[1]{https://doi.org/#1}

\bibitem{autoqrepo}
\autoqq: An automata-based {C++} tool for quantum program verification (Jan
  2024), \url{https://github.com/alan23273850/AutoQ/tree/CAV24}

\bibitem{popl}
Abdulla, P.A., Chen, Y.G., Chen, Y.F., Holík, L., Lengal, O., Lo, F.Y., Lin,
  J.A., Tsai, W.L.: Verifying quantum circuits with level-synchronized tree
  automata. under submission  (2024)

\bibitem{AbdullaCHMV10}
Abdulla, P.A., Chen, Y., Hol{\'{\i}}k, L., Mayr, R., Vojnar, T.: When
  simulation meets antichains. In: Esparza, J., Majumdar, R. (eds.) Tools and
  Algorithms for the Construction and Analysis of Systems, 16th International
  Conference, {TACAS} 2010, Held as Part of the Joint European Conferences on
  Theory and Practice of Software, {ETAPS} 2010, Paphos, Cyprus, March 20-28,
  2010. Proceedings. Lecture Notes in Computer Science, vol.~6015, pp.
  158--174. Springer (2010). \doi{10.1007/978-3-642-12002-2\_14},
  \url{https://doi.org/10.1007/978-3-642-12002-2\_14}

\bibitem{Aharonov03}
Aharonov, D.: A simple proof that {Toffoli} and {Hadamard} are quantum
  universal (2003). \doi{10.48550/arxiv.quant-ph/0301040},
  \url{https://arxiv.org/abs/quant-ph/0301040}

\bibitem{KeY}
Ahrendt, W., Beckert, B., Bubel, R., H{\"{a}}hnle, R., Schmitt, P.H., Ulbrich,
  M. (eds.): Deductive Software Verification - The {KeY} Book - From Theory to
  Practice, Lecture Notes in Computer Science, vol. 10001. Springer (2016).
  \doi{10.1007/978-3-319-49812-6},
  \url{https://doi.org/10.1007/978-3-319-49812-6}

\bibitem{amy2018towards}
Amy, M.: Towards large-scale functional verification of universal quantum
  circuits. In: Quantum Physics and Logic (2018)

\bibitem{weakly:chris}
Andrés-Martínez, P., Heunen, C.: Weakly measured while loops: peeking at
  quantum states. Quantum Science and Technology  \textbf{7}(2),  025007 (feb
  2022). \doi{10.1088/2058-9565/ac47f1},
  \url{https://dx.doi.org/10.1088/2058-9565/ac47f1}

\bibitem{AnticoliPTZ16}
Anticoli, L., Piazza, C., Taglialegne, L., Zuliani, P.: Towards quantum
  programs verification: From {Quipper} circuits to {QPMC}. In: Devitt, S.J.,
  Lanese, I. (eds.) Reversible Computation - 8th International Conference, {RC}
  2016, Bologna, Italy, July 7-8, 2016, Proceedings. Lecture Notes in Computer
  Science, vol.~9720, pp. 213--219. Springer (2016).
  \doi{10.1007/978-3-319-40578-0\_16},
  \url{https://doi.org/10.1007/978-3-319-40578-0\_16}

\bibitem{BaudinBBCKKMPPS21}
Baudin, P., Bobot, F., B{\"{u}}hler, D., Correnson, L., Kirchner, F., Kosmatov,
  N., Maroneze, A., Perrelle, V., Prevosto, V., Signoles, J., Williams, N.: The
  dogged pursuit of bug-free {C} programs: the {Frama-C} software analysis
  platform. Commun. {ACM}  \textbf{64}(8),  56--68 (2021).
  \doi{10.1145/3470569}

\bibitem{bertot2013interactive}
Bertot, Y., Cast{\'e}ran, P.: Interactive theorem proving and program
  development: Coq’Art: the calculus of inductive constructions. Springer
  Science \& Business Media (2013)

\bibitem{RUS2}
Bocharov, A., Roetteler, M., Svore, K.M.: Efficient synthesis of universal
  repeat-until-success quantum circuits. Phys. Rev. Lett.  \textbf{114},
  080502 (Feb 2015). \doi{10.1103/PhysRevLett.114.080502},
  \url{https://link.aps.org/doi/10.1103/PhysRevLett.114.080502}

\bibitem{BouajjaniHHTV08}
Bouajjani, A., Habermehl, P., Hol{\'{\i}}k, L., Touili, T., Vojnar, T.:
  Antichain-based universality and inclusion testing over nondeterministic
  finite tree automata. In: Ibarra, O.H., Ravikumar, B. (eds.) Implementation
  and Applications of Automata, 13th International Conference, {CIAA} 2008, San
  Francisco, California, USA, July 21-24, 2008. Proceedings. Lecture Notes in
  Computer Science, vol.~5148, pp. 57--67. Springer (2008).
  \doi{10.1007/978-3-540-70844-5\_7},
  \url{https://doi.org/10.1007/978-3-540-70844-5\_7}

\bibitem{BoykinMPRV00}
Boykin, P.O., Mor, T., Pulver, M., Roychowdhury, V.P., Vatan, F.: A new
  universal and fault-tolerant quantum basis. Inf. Process. Lett.
  \textbf{75}(3),  101--107 (2000). \doi{10.1016/S0020-0190(00)00084-3},
  \url{https://doi.org/10.1016/S0020-0190(00)00084-3}

\bibitem{burgholzer2020advanced}
Burgholzer, L., Wille, R.: Advanced equivalence checking for quantum circuits.
  IEEE Transactions on Computer-Aided Design of Integrated Circuits and Systems
   \textbf{40}(9),  1810--1824 (2020)

\bibitem{chareton2021automated}
Chareton, C., Bardin, S., Bobot, F., Perrelle, V., Valiron, B.: An automated
  deductive verification framework for circuit-building quantum programs. In:
  Programming Languages and Systems: 30th European Symposium on Programming,
  ESOP 2021, Held as Part of the European Joint Conferences on Theory and
  Practice of Software, ETAPS 2021, Luxembourg City, Luxembourg, March
  27--April 1, 2021, Proceedings 30. pp. 148--177. Springer International
  Publishing (2021)

\bibitem{Chareton2021}
Chareton, C., Bardin, S., Bobot, F., Perrelle, V., Valiron, B.: An automated
  deductive verification framework for circuit-building quantum programs. In:
  Yoshida, N. (ed.) ESOP. Lecture Notes in Computer Science, vol. 12648, pp.
  148--177. {Springer International Publishing}, {Cham} (March 2021)

\bibitem{ChenCJJL24}
Chen, T., Chen, Y., Jiang, J.R., Lengál, O., Jobranová, S.: Accelerating
  quantum circuit simulation with symbolic execution and loop summarization.
  In: Proc.\ of ICCAD'24. ACM (2024)

\bibitem{9951285}
Chen, T.F., Jiang, J.H.R., Hsieh, M.H.: Partial equivalence checking of quantum
  circuits. In: 2022 IEEE International Conference on Quantum Computing and
  Engineering (QCE). pp. 594--604 (2022). \doi{10.1109/QCE53715.2022.00082}

\bibitem{chen2023autoq}
Chen, Y.F., Chung, K.M., Leng{\'a}l, O., Lin, J.A., Tsai, W.L.: \textsc{AutoQ}:
  An automata-based quantum circuit verifier. In: International Conference on
  Computer Aided Verification. pp. 139--153. Springer (2023)

\bibitem{ChenCLLTY23}
Chen, Y., Chung, K., Leng\'{a}l, O., Lin, J., Tsai, W., Yen, D.: An
  automata-based framework for verification and bug hunting in quantum
  circuits. In: 44th ACM SIGPLAN Conference on Programming Language Design and
  Implementation---PLDI'23. ACM (2023)

\bibitem{Coecke2011}
Coecke, B., Duncan, R.: Interacting quantum observables: categorical algebra
  and diagrammatics. New Journal of Physics  \textbf{13}(4),  043016 (apr
  2011). \doi{10.1088/1367-2630/13/4/043016}

\bibitem{CohenDHLMSST09}
Cohen, E., Dahlweid, M., Hillebrand, M.A., Leinenbach, D., Moskal, M., Santen,
  T., Schulte, W., Tobies, S.: {VCC:} {A} practical system for verifying
  concurrent {C}. In: Berghofer, S., Nipkow, T., Urban, C., Wenzel, M. (eds.)
  Theorem Proving in Higher Order Logics, 22nd International Conference, TPHOLs
  2009, Munich, Germany, August 17-20, 2009. Proceedings. Lecture Notes in
  Computer Science, vol.~5674, pp. 23--42. Springer (2009).
  \doi{10.1007/978-3-642-03359-9\_2},
  \url{https://doi.org/10.1007/978-3-642-03359-9\_2}

\bibitem{de2008z3}
De~Moura, L., Bj{\o}rner, N.: Z3: An efficient {SMT} solver. In: Tools and
  Algorithms for the Construction and Analysis of Systems: 14th International
  Conference, TACAS 2008, Held as Part of the Joint European Conferences on
  Theory and Practice of Software, ETAPS 2008, Budapest, Hungary, March
  29-April 6, 2008. Proceedings 14. pp. 337--340. Springer (2008)

\bibitem{d2006quantum}
D'Hondt, E., Panangaden, P.: Quantum weakest preconditions. Mathematical
  Structures in Computer Science  \textbf{16}(3),  429--451 (2006)

\bibitem{Diffrule2}
Fang, W., Ying, M., Wu, X.: Differentiable quantum programming with unbounded
  loops. ACM Trans. Softw. Eng. Methodol.  \textbf{33}(1) (nov 2023).
  \doi{10.1145/3617178}, \url{https://doi.org/10.1145/3617178}

\bibitem{feng2021quantum}
Feng, Y., Ying, M.: Quantum {Hoare} logic with classical variables. ACM
  Transactions on Quantum Computing  \textbf{2}(4),  1--43 (2021)

\bibitem{FengYY13}
Feng, Y., Yu, N., Ying, M.: Model checking quantum {Markov} chains. J. Comput.
  Syst. Sci.  \textbf{79}(7),  1181--1198 (2013).
  \doi{10.1016/j.jcss.2013.04.002},
  \url{https://doi.org/10.1016/j.jcss.2013.04.002}

\bibitem{filliatre2013why3}
Filli{\^a}tre, J.C., Paskevich, A.: Why3—where programs meet provers. In:
  Programming Languages and Systems: 22nd European Symposium on Programming,
  ESOP 2013, Held as Part of the European Joint Conferences on Theory and
  Practice of Software, ETAPS 2013, Rome, Italy, March 16-24, 2013. Proceedings
  22. pp. 125--128. Springer (2013)

\bibitem{Floyd93}
Floyd, R.W.: Assigning meanings to programs. In: Mathematical Aspects of
  Computer Science. pp. 19--32. Proceedings of Symposia in Applied Mathematics,
  AMS (1967), \url{https://mathscinet.ams.org/mathscinet/article?mr=235771}

\bibitem{Grover96}
Grover, L.K.: A fast quantum mechanical algorithm for database search. In:
  Miller, G.L. (ed.) Proceedings of the Twenty-Eighth Annual {ACM} Symposium on
  the Theory of Computing, Philadelphia, Pennsylvania, USA, May 22-24, 1996.
  pp. 212--219. {ACM} (1996). \doi{10.1145/237814.237866},
  \url{https://doi.org/10.1145/237814.237866}

\bibitem{HahnleH19}
H{\"{a}}hnle, R., Huisman, M.: Deductive software verification: From
  pen-and-paper proofs to industrial tools. In: Steffen, B., Woeginger, G.J.
  (eds.) Computing and Software Science - State of the Art and Perspectives,
  Lecture Notes in Computer Science, vol. 10000, pp. 345--373. Springer (2019).
  \doi{10.1007/978-3-319-91908-9\_18},
  \url{https://doi.org/10.1007/978-3-319-91908-9\_18}

\bibitem{hietala2019verified}
Hietala, K., Rand, R., Hung, S.H., Wu, X., Hicks, M.: Verified optimization in
  a quantum intermediate representation. arXiv preprint arXiv:1904.06319
  (2019)

\bibitem{Hoare69}
Hoare, C.A.R.: An axiomatic basis for computer programming. Commun. {ACM}
  \textbf{12}(10),  576--580 (1969). \doi{10.1145/363235.363259},
  \url{https://doi.org/10.1145/363235.363259}

\bibitem{HolikLSV11}
Hol{\'{\i}}k, L., Leng{\'{a}}l, O., Sim{\'{a}}cek, J., Vojnar, T.: Efficient
  inclusion checking on explicit and semi-symbolic tree automata. In: Bultan,
  T., Hsiung, P. (eds.) Automated Technology for Verification and Analysis, 9th
  International Symposium, {ATVA} 2011, Taipei, Taiwan, October 11-14, 2011.
  Proceedings. Lecture Notes in Computer Science, vol.~6996, pp. 243--258.
  Springer (2011). \doi{10.1007/978-3-642-24372-1\_18},
  \url{https://doi.org/10.1007/978-3-642-24372-1\_18}

\bibitem{JacobsSP10}
Jacobs, B., Smans, J., Piessens, F.: A quick tour of the {VeriFast} program
  verifier. In: Ueda, K. (ed.) Programming Languages and Systems - 8th Asian
  Symposium, {APLAS} 2010, Shanghai, China, November 28 - December 1, 2010.
  Proceedings. Lecture Notes in Computer Science, vol.~6461, pp. 304--311.
  Springer (2010). \doi{10.1007/978-3-642-17164-2\_21},
  \url{https://doi.org/10.1007/978-3-642-17164-2\_21}

\bibitem{Leino10}
Leino, K.R.M.: {Dafny}: An automatic program verifier for functional
  correctness. In: Clarke, E.M., Voronkov, A. (eds.) Logic for Programming,
  Artificial Intelligence, and Reasoning - 16th International Conference,
  LPAR-16, Dakar, Senegal, April 25-May 1, 2010, Revised Selected Papers.
  Lecture Notes in Computer Science, vol.~6355, pp. 348--370. Springer (2010).
  \doi{10.1007/978-3-642-17511-4\_20},
  \url{https://doi.org/10.1007/978-3-642-17511-4\_20}

\bibitem{lengal2012vata}
Leng{\'a}l, O., {\v{S}}im{\'a}{\v{c}}ek, J., Vojnar, T.: {VATA}: A library for
  efficient manipulation of non-deterministic tree automata. In: International
  Conference on Tools and Algorithms for the Construction and Analysis of
  Systems. pp. 79--94. Springer (2012)

\bibitem{liu2019formal}
Liu, J., Zhan, B., Wang, S., Ying, S., Liu, T., Li, Y., Ying, M., Zhan, N.:
  Formal verification of quantum algorithms using quantum {Hoare} logic. In:
  International conference on computer aided verification. pp. 187--207.
  Springer (2019)

\bibitem{MateusRSS09}
Mateus, P., Ramos, J., Sernadas, A., Sernadas, C.: Temporal Logics for
  Reasoning about Quantum Systems, p. 389–413. Cambridge University Press
  (2009). \doi{10.1017/CBO9781139193313.011}

\bibitem{RUS}
Paetznick, A., Svore, K.M.: Repeat-until-success: Non-deterministic
  decomposition of single-qubit unitaries. Quantum Info. Comput.
  \textbf{14}(15–16),  1277–1301 (nov 2014)

\bibitem{perdrix2008quantum}
Perdrix, S.: Quantum entanglement analysis based on abstract interpretation.
  In: International Static Analysis Symposium. pp. 270--282. Springer (2008)

\bibitem{unruh2019quantum}
Unruh, D.: Quantum {Hoare} logic with ghost variables. In: 2019 34th Annual
  ACM/IEEE Symposium on Logic in Computer Science (LICS). pp. 1--13. IEEE
  (2019)

\bibitem{10.1145/3489517.3530481}
Wei, C.Y., Tsai, Y.H., Jhang, C.S., Jiang, J.H.R.: Accurate bdd-based unitary
  operator manipulation for scalable and robust quantum circuit verification.
  In: Proceedings of the 59th ACM/IEEE Design Automation Conference. p.
  523–528. DAC '22, Association for Computing Machinery, New York, NY, USA
  (2022). \doi{10.1145/3489517.3530481},
  \url{https://doi.org/10.1145/3489517.3530481}

\bibitem{wenzel2008isabelle}
Wenzel, M., Paulson, L.C., Nipkow, T.: The {Isabelle} framework. In: Theorem
  Proving in Higher Order Logics: 21st International Conference, TPHOLs 2008,
  Montreal, Canada, August 18-21, 2008. Proceedings 21. pp. 33--38. Springer
  (2008)

\bibitem{XuFMD22}
Xu, M., Fu, J., Mei, J., Deng, Y.: Model checking {QCTL} plus on quantum
  {Markov} chains. Theor. Comput. Sci.  \textbf{913},  43--72 (2022).
  \doi{10.1016/j.tcs.2022.01.044},
  \url{https://doi.org/10.1016/j.tcs.2022.01.044}

\bibitem{xu2022quartz}
Xu, M., Li, Z., Padon, O., Lin, S., Pointing, J., Hirth, A., Ma, H., Palsberg,
  J., Aiken, A., Acar, U.A., et~al.: Quartz: superoptimization of quantum
  circuits. In: Proceedings of the 43rd ACM SIGPLAN International Conference on
  Programming Language Design and Implementation. pp. 625--640 (2022)

\bibitem{ying2012floyd}
Ying, M.: Floyd-{Hoare} logic for quantum programs. ACM Transactions on
  Programming Languages and Systems (TOPLAS)  \textbf{33}(6),  1--49 (2012)

\bibitem{yu2021quantum}
Yu, N., Palsberg, J.: Quantum abstract interpretation. In: Proceedings of the
  42nd ACM SIGPLAN International Conference on Programming Language Design and
  Implementation. pp. 542--558 (2021)

\bibitem{zhou2023coqq}
Zhou, L., Barthe, G., Strub, P.Y., Liu, J., Ying, M.: Coqq: Foundational
  verification of quantum programs. Proceedings of the ACM on Programming
  Languages  \textbf{7}(POPL),  833--865 (2023)

\bibitem{zhou2019applied}
Zhou, L., Yu, N., Ying, M.: An applied quantum {Hoare} logic. In: Proceedings
  of the 40th ACM SIGPLAN Conference on Programming Language Design and
  Implementation. pp. 1149--1162 (2019)

\bibitem{Diffrule1}
Zhu, S., Hung, S.H., Chakrabarti, S., Wu, X.: On the principles of
  differentiable quantum programming languages. In: Proceedings of the 41st ACM
  SIGPLAN Conference on Programming Language Design and Implementation. p.
  272–285. PLDI 2020, Association for Computing Machinery, New York, NY, USA
  (2020). \doi{10.1145/3385412.3386011},
  \url{https://doi.org/10.1145/3385412.3386011}

\end{thebibliography}
%%%%%%%%%%%%%%%%%%%%%%%%%%%%%%%%%%%%%%%%%%%%%%%%%%%%%%%%%%%%%%%%%%

\appendix
\clearpage

\end{document}